\documentclass{llncs}
\usepackage{amssymb,amsmath,enumerate,graphics,graphicx,color,enumerate}
\usepackage{latexsym, subcaption}
\usepackage{tabularx}
\usepackage{tikz}
\usetikzlibrary{shapes.geometric, arrows.meta, calc, positioning}
\usepackage{boxedminipage}
\usepackage{afterpage}
\usepackage{comment}
\usepackage[11pt]{moresize}

\DeclareMathOperator{\dist}{dist}

\newtheorem{nclaim}{Claim}

\newcommand{\NP}{{\sf NP}}
\newcommand{\PP}{{\sf P}}

\newcommand{\shc}{{\sc Surjective $H$-Colouring}}

\newcommand{\problemdef}[3]{
	\begin{center}
		\begin{boxedminipage}{.99\textwidth}
			\textsc{{#1}}\\[2pt]
			\begin{tabular}{ r p{0.8\textwidth}}
				\textit{~~~~Instance:} & {#2}\\
				\textit{Question:} & {#3}
			\end{tabular}
		\end{boxedminipage}
	\end{center}
}

\newcommand{\dee}{{\mathcal{D}}}

\newcommand{\disjun}{{\cup}}

\newcommand{\convexpath}[2]{
  [   
  create hullcoords/.code={
    \global\edef\namelist{#1}
    \foreach [count=\counter] \nodename in \namelist {
      \global\edef\numberofnodes{\counter}
      \coordinate (hullcoord\counter) at (\nodename);
    }
    \coordinate (hullcoord0) at (hullcoord\numberofnodes);
    \pgfmathtruncatemacro\lastnumber{\numberofnodes+1}
    \coordinate (hullcoord\lastnumber) at (hullcoord1);
  },
  create hullcoords
  ]
  ($(hullcoord1)!#2!-90:(hullcoord0)$)
  \foreach [
  evaluate=\currentnode as \previousnode using \currentnode-1,
  evaluate=\currentnode as \nextnode using \currentnode+1
  ] \currentnode in {1,...,\numberofnodes} {
    let \p1 = ($(hullcoord\currentnode) - (hullcoord\previousnode)$),
    \n1 = {atan2(\y1,\x1) + 90},
    \p2 = ($(hullcoord\nextnode) - (hullcoord\currentnode)$),
    \n2 = {atan2(\y2,\x2) + 90},
    \n{delta} = {Mod(\n2-\n1,360) - 360}
    in 
    {arc [start angle=\n1, delta angle=\n{delta}, radius=#2]}
    -- ($(hullcoord\nextnode)!#2!-90:(hullcoord\currentnode)$) 
  }
}

\let\prooforig\proof
\let\prooforigend\endproof
\renewenvironment{proof}{\prooforig}{\qed \prooforigend}

\begin{document}

\tikzset{
	every loop/.style={min distance=10mm,looseness=10}
}
\tikzstyle{problem} = [rectangle, rounded corners, minimum width=3cm, minimum height=1cm,text centered, draw=black]
\tikzstyle{arrow_back} = [thick,<-,>={Stealth[length=3mm]}]

\title{Surjective $H$-Colouring: New Hardness Results\thanks{Supported by the Research Council of Norway via the project ``CLASSIS'' and the Leverhulme Trust (RPG-2016-258).}}

\author{Petr A. Golovach$^1$ \and
	Matthew Johnson$^2$ \and
	Barnaby Martin$^2$\and \\ 
	Dani\"el Paulusma$^2$
	\and
    Anthony Stewart$^2$}

\institute{$^1$Department of Informatics, University of Bergen,
Bergen, Norway\\
\texttt{petr.golovach@ii.uib.no}\\
$^2$School of Engineering and Computing Sciences,
Durham University,\\
South Road, Durham, DH1 3LE, U.K.\\
\texttt{$\{$matthew.johnson2,barnaby.d.martin,daniel.paulusma,a.g.stewart$\}$@durham.ac.uk}}

\maketitle

\begin{abstract}
A homomorphism from a graph $G$ to a graph $H$ is a vertex mapping  
$f$ from the vertex set of $G$ to the vertex set of $H$
such
that there is an edge between vertices $f(u)$ and $f(v)$ of $H$ whenever there is an edge between vertices $u$ and $v$ of $G$.
The {\sc $H$-Colouring} problem is to decide if a graph~$G$ allows a homomorphism to a fixed graph~$H$.
We continue a study on a variant of this problem, namely
the \shc{} problem, which imposes the homomorphism to be vertex-surjective.
We build upon previous results and show that this problem is \NP-complete for every connected graph $H$ that has exactly two vertices with a self-loop as long as these two vertices are not adjacent. 
As a result, we can classify the computational complexity of \shc{} for every graph $H$ on at most four vertices. 
\end{abstract}

\section{Introduction}

The well-known {\sc Colouring} problem is to decide if the vertices of a given graph can be properly coloured with 
at most $k$ colours for some given integer $k$.
If we exclude $k$ from the input and assume it is fixed,
we obtain the {\sc $k$-Colouring} problem. 
A {\it homomorphism} from a graph $G=(V_G, E_G)$ to a graph~$H=(V_H, E_H)$ is a vertex mapping $f \, : \, V_G \rightarrow V_H$, such
that there is an edge between $f(u)$ and $f(v)$ in $H$
whenever there is an edge between $u$ and $v$ in $G$.
We observe that {\sc $k$-Colouring} is equivalent to the problem of asking if a graph allows a homomorphism to the complete graph $K_k$ on $k$ vertices.
Hence, a natural generalization of the {\sc $k$-Colouring} problem is the 
$H$-{\sc Colouring} problem, which asks if a given graph allows a homomorphism to an arbitrary fixed graph $H$. We call this fixed graph~$H$ the {\it target graph}.
Throughout the paper we consider undirected graphs with no multiple edges. We assume that an input graph~$G$ contains no vertices with self-loops (we call such vertices 
{\it reflexive}), whereas a target graph~$H$ may contain such vertices. 
We call $H$ {\it reflexive} if all its vertices are reflexive, and {\it irreflexive} if all its vertices are irreflexive.

For a survey on graph homomorphisms we refer the reader to the textbook of Hell and Ne\v{s}et\v{r}il \cite{HN04}. Here,
we will discuss the $H$-{\sc Colouring} problem, a number of its variants and their relations to each other. In particular, we will focus on the {\it surjective} variant: 
a homomorphism~$f$ from a graph $G$ to a graph~$H$ is {\it (vertex-)surjective} if $f$ is surjective, that is, if for every vertex $x\in V_H$ there exists at least one vertex~$u\in V_G$ with $f(u)=x$.

The computational complexity of $H$-{\sc Colouring} has been determined completely. The problem is trivial if $H$ contains a reflexive vertex~$u$ (we can map each vertex of the input graph to $u$). If $H$ has no reflexive vertices, then
the Hell-Ne\v{s}et\v{r}il dichotomy theorem \cite{HN90} tells us that $H$-Colouring is solvable in polynomial time if $H$ is bipartite and that it is \NP-complete otherwise.

The {\sc List $H$-Colouring} problem takes as input a graph $G$ and a function $L$ that assigns to each $u\in V_G$ 
a list $L(u)\subseteq V_H$. The question is whether $G$ allows a homomorphism~$f$ to the target $H$ with $f(u)\in L(u)$ for every $u\in V_G$.
Feder, Hell and Huang~\cite{FHH03} proved that  {\sc List $H$-Colouring} is polynomial-time solvable if $H$ is a 
bi-arc graph and \NP-complete otherwise (we refer to~\cite{FHH03} for the definition of a bi-arc graph).
A homomorphism $f$ from $G$ to an induced subgraph $H$ of $G$
is a {\it retraction} if $f(x)=x$ for every $x\in V_H$, and  
we say that $G$ {\it retracts to} $H$. 
A retraction from $G$ to
$H$ can be viewed as a list-homomorphism: choose
$L(u)=\{u\}$ if $u\in V_H$, and $L(u)=V_H$ if $u\in V_G\setminus V_H$.
The corresponding decision problem is called $H$-{\sc Retraction}.
The computational complexity of $H$-{\sc Retraction} has not yet been classified.
Feder et al.~\cite{FHJKN10} determined the complexity of the $H$-{\sc Retraction} problem whenever $H$ is a 
pseudo-forest (a graph in which every connected component has at most one cycle).
They also showed that $H$-{\sc Retraction} is \NP-complete if $H$ contains a connected component in which the reflexive vertices induce a disconnected graph.

As mentioned, 
we impose a (vertex-)surjectivity condition on the graph homomorphism.
Such a condition can be imposed locally or globally.  If we require a homomorphism~$f$ 
from a graph $G$ to a graph $H$ 
to be surjective when restricted to the 
open neighbourhood of every vertex $u$ of $G$, we say that $f$ is an $H$-{\it role assignment}. 
The corresponding decision problem is called $H$-{\sc Role Assignment} and its computational complexity has been fully classified~\cite{FP05}.
We refer to the survey of Fiala and Kratochv\'il~\cite{FK08} for further details on locally constrained homomorphisms and from here on only consider global surjectivity.

It has been shown 
that deciding whether a given graph $G$ allows a surjective homomorphism to a given graph $H$ is \NP-complete even if $G$ and $H$ both belong to one of the following graph classes: disjoint unions of paths; disjoint unions of complete graphs; trees; connected cographs; connected proper interval graphs; and connected split graphs~\cite{GLMP12}.  Hence it is natural, just as before, to fix $H$, which yields the following problem:
\problemdef{\shc{}}{a graph $G$.}{does there exist a surjective homomorphism from $G$ to $H$?}
We emphasize that we are considering vertex-surjectivity and that 
being vertex-surjective is a different
condition than being edge-surjective.
A homomorphism from a graph $G$ to a graph $H$ is called {\it edge-surjective} or a {\it compaction} if for any edge $xy\in E_H$ with $x\neq y$ there exists an edge $uv\in E_G$ with $f(u)=x$ and $f(v)=y$. Note that the edge-surjectivity condition does not hold for any self-loops $xx\in E_H$.
If $f$ is a compaction from $G$ to $H$, we say that $G$ {\it compacts} to $H$.
The corresponding decision problem is known as the $H$-{\sc Compaction} problem. A full classification of this problem is still wide open.
However partial results are known, for example when $H$ is a reflexive cycle,  an irreflexive cycle, or a graph on at most four vertices~\cite{Vi02,Vi04,Vi05}, or when $G$ is restricted to some special graph class~\cite{Vi13}.
Vikas also showed that whenever {\sc $H$-Retraction} is polynomial-time solvable, then so is {\sc $H$-Compaction}~\cite{Vi04}.
Whether the reverse implication holds is not known.
A complete complexity classification of \shc{} is also still open. Below we survey the known results.

We first consider irreflexive target graphs~$H$.
The {\sc Surjective $H$-Colouring} problem is \NP-complete for every such graph~$H$ if $H$ is non-bipartite, as observed by Golovach et al.~\cite{GPS12}. The straightforward reduction is from the corresponding {\sc $H$-Colouring} problem, which is \NP-complete due to the aforementioned Hell-Ne\v{s}et\v{r}il dichotomy theorem.
However, the complexity classifications of {\sc $H$-Colouring} and {\sc Surjective $H$-Colouring} do not coincide: there exist bipartite graphs $H$ for which \shc{} is \NP-complete,  for instance when $H$ is the graph obtained from a 6-vertex cycle to each of which vertices we add a path of length 3~\cite{BKM12},
or when $H$ is the 6-vertex cycle itself~\cite{Vi17}.

We now consider target graphs with at least one reflexive vertex.
Unlike the {\sc $H$-Colouring} problem, the presence of a reflexive vertex does not make the \shc{} problem trivial to solve.
We call a connected graph {\it loop-connected} if all its reflexive vertices induce a connected subgraph.
Golovach, Paulusma and Song \cite{GPS12} showed that if $H$ is a tree (in this context, a connected graph with no cycles of length at least~3) then \shc{} is polynomial-time solvable if $H$ is loop-connected and \NP-complete otherwise.  As such the following question is natural:

\smallskip
\noindent
{\it Is \shc{} \NP-complete for every connected graph $H$ that is not loop-connected?}

\smallskip
\noindent
The reverse statement is not true (if \PP $\neq$ \NP): \shc{} is \NP-complete when $H$ is the 4-vertex cycle $C_4^*$ with a self-loop in each of its vertices. This result has been shown by Martin and Paulusma~\cite{MP15} and independently by Vikas, as announced in~\cite{Vi13}. 
Recall also that \shc{} is \NP-complete if $H$ is irreflexive (and thus loop-connected) and non-bipartite. 

It is known that \shc{} is polynomial-time solvable whenever $H$-{\sc Compaction} is~\cite{BKM12}. Recall that $H$-{\sc Compaction} is polynomial-time solvable whenever $H$-{\sc Retraction} is~\cite{Vi04}.
Hence, for instance, the aforementioned result of Feder, Hell and Huang~\cite{FHH03} implies that \shc{} is polynomial-time solvable if $H$ is a bi-arc graph.  
We also recall that {\sc $H$-Retraction} is \NP-complete whenever $H$ is a connected graph that is not loop-connected~\cite{FHJKN10}. Hence, an affirmative answer to the above question would mean that for these target graphs~$H$ the complexities of $H$-{\sc Retraction},
$H$-{\sc Compaction} and \shc{} coincide. 

In Figure~\ref{shc:fig:comp-rel} we display the relationships between the different problems discussed. In particular,
it is a major open problem whether the computational complexities of $H$-{\sc Compaction}, $H$-{\sc Retraction} and \shc{} coincide for each target graph~$H$. 
Even showing this for specific cases, such as the case $H=C_4^*$, has been proven to be non-trivial.
If it is true, it would relate the \shc{} problem to  
a well-known conjecture of Feder and Vardi~\cite{FV98}, which states that 
the ${\cal H}$-{\sc Constraint Satisfaction} problem has a dichotomy when ${\cal H}$ is some fixed finite target structure and which is equivalent to conjecturing that $H$-{\sc Retraction} has a dichotomy~\cite{FV98}.
We refer to the survey of Bodirsky, Kara and Martin~\cite{BKM12} for more details on the \shc{} problem from a constraint satisfaction point of view.

\begin{figure}
\begin{centering}
\newcommand{\separation}{3mm}
\begin{tikzpicture}[problem/.style={draw,rectangle,rounded corners,minimum width=1cm,minimum height=0.5cm,font=\ssmall}]
\node[problem] (listHcolouring) {{\sc List $H$-Colouring}};
\node[problem, right=\separation of listHcolouring] (Hretraction)  {{\sc $H$-Retraction}};
\node[problem, right=\separation of Hretraction] (Hcompaction)  {{\sc $H$-Compaction}};
\node[problem, right=\separation of Hcompaction] (surjectiveHcolouring) {{\sc Surj $H$-Colouring}};
\node[problem, right=\separation of surjectiveHcolouring] (Hcolouring)  {{\sc $H$-Colouring}};
\draw[->, thick] (listHcolouring) to (Hretraction);
\draw[->, thick] (Hretraction) to (Hcompaction);
\draw[->, thick] (Hcompaction) to (surjectiveHcolouring);
\draw[->, thick] (surjectiveHcolouring) to (Hcolouring);
\end{tikzpicture}
\caption{Relations between $H$-{\sc Colouring} and its variants. An arrow from one problem to another indicates that the latter problem is polynomial-time solvable for a target graph~$H$ if the former is polynomial-time solvable for $H$. Reverse arrows do not hold for the leftmost and rightmost arrows,
as witnessed by the reflexive 4-vertex cycle for the rightmost arrow and by any reflexive tree that is not a reflexive interval graph for the leftmost arrow (Feder, Hell and Huang~\cite{FHH03} showed that the only reflexive bi-arc graphs are reflexive interval graphs).
It is not known if the reverse direction holds for the two middle arrows.}\label{shc:fig:comp-rel}
\end{centering}	
\end{figure}
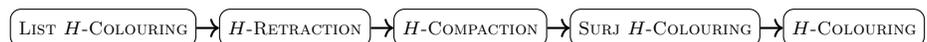

\subsection{Our Results}

We present further progress on the research question of whether \shc{} is \NP-complete for every connected graph $H$ that is not loop-connected. 
We first consider the case where the target graph $H$ is a connected graph with exactly two reflexive vertices that are non-adjacent. In Section~\ref{s-refver} we prove that \shc{} is indeed \NP-complete for every such target graph~$H$.
In the same section we slightly generalize this result by showing that it holds even if the reflexive vertices of $H$ can be partitioned into two non-adjacent sets of twin vertices. This enables us to 
classify in Section~\ref{s-four} the computational complexity of \shc{} for every graph $H$ on at most four vertices, just as Vikas~\cite{Vi05} did for the {\sc $H$-Compaction} problem. 
A classification of {\sc Surjective $H$-Colouring} for target graphs~$H$ on at most four vertices has also been announced by Vikas in~\cite{Vi13}. As we will illustrate for one particular case,
it is interesting to note that \NP-hardness proofs for {\sc $H$-Compaction} of~\cite{Vi05} may lift to \NP-hardness for {\sc Surjective $H$-Colouring}.
However, this is not true for the reflexive cycle $C^*_4$, where a totally new proof was required.

\subsection{Future Work}

To conjecture a dichotomy of \shc{} between \PP\ and \NP-complete seems still to be difficult.
Our first goal is to prove 
that \shc{} is \NP-complete for every connected graph $H$ that is not loop-connected. However, doing this via using our current techniques does not seem straightforward  
and we may need new hardness reductions. 
Another way forward is to prove polynomial equivalence between the three problems \shc{}, $H$-{\sc Compaction} and $H$-{\sc Retraction}. 
However, completely achieving this goal also seems far from trivial.
Our classification for target graphs $H$ up to four vertices does show such an equivalence for these cases (see Section~\ref{s-four}). 

\section{Two Non-Adjacent Reflexive Vertices}\label{s-refver}

We say that a graph is  {\it $2$-reflexive} if it contains exactly $2$ reflexive vertices {\it that are non-adjacent}. 
In this section we will prove that \shc{} is \NP-complete whenever $H$ is connected and 2-reflexive. The problem is readily seen to be in \NP. Our \NP-hardness reduction uses similar ingredients as the reduction of Golovach, Paulusma and Song~\cite{GPS12} for proving \NP-hardness when $H$ is a tree that is not loop-connected.
There are, however, a number of differences.
For instance, we will reduce from a factor cut problem instead of the less general matching cut problem used in~\cite{GPS12}. We will explain these two problems and prove \NP-hardness for the former one in 
Section~\ref{s-factor}. Then in Section~\ref{s-hard} we give our hardness reduction, and in Section~\ref{s-refcl} we extend our result to be valid 
for target graphs~$H$ with more than two reflexive vertices as long as these reflexive vertices can be partitioned into two non-adjacent sets of  twin vertices. 

\subsection{Factor Cuts}\label{s-factor}

Let $G=(V_G, E_G)$ be a connected graph.  For $v \in V_G$ and $E \subseteq E_G$, let $d_E(v)$ denote the number of edges of $E$ incident with $v$.
For a partition $(V_1,V_2)$ of $V_G$, let $E_G(V_1,V_2)$ denote the set of edges between $V_1$ and $V_2$ in $G$.  

Let $i$ and $j$ be positive integers, $i\leq j$.  Let $(V_1,V_2)$ be a partition of $V_G$ and let $M=E_G(V_1,V_2)$.  Then $(V_1,V_2)$ is an \emph{$(i,j)$-factor cut} of $G$ if, for all $v \in V_1$, $d_M(v) \leq i$, and, for all $v \in V_2$, $d_M(v) \leq j$.  
Observe that if a vertex $v$ exists with degree at most $j$, then there is a trivial $(i, j)$-factor cut $(V \setminus \{v\}, \{v\})$.
Two distinct vertices $s$ and $t$ in $V_G$ are {\it $(i,j)$-factor roots} of $G$ if, for each $(i,j)$-factor cut $(V_1,V_2)$ of $G$, $s$ and $t$ belong to different parts of the partition and, if $i < j$, $s \in V_1$ and $t \in V_2$ (of course, if $i=j$, we do not require the latter condition as $(V_2,V_1)$ is also an $(i,j)$-factor cut).
We note that when no $(i,j)$-factor cut exists, every pair of vertices is a pair of $(i,j)$-factor roots.  We define the following decision problem.

\problemdef{$(i,j)$-Factor Cut with Roots}{a connected graph $G$ with roots $s$~and~$t$.}{does $G$ have an $(i,j)$-factor cut?}
We emphasize that the $(i,j)$-factor roots are given as part of the input.  That is, the problem asks whether or not an $(i,j)$-factor cut $(V_1,V_2)$ exists, but we know already that if it does, 
then $s$ and $t$ belong to different parts of the partition.
That is, we actually define {\sc $(i,j)$-Factor Cut with Roots} to be a promise problem in which we assume that if an $(i,j)$-factor cut exists then it has the property that $s$ and $t$ belong to different parts of the partition. The promise class may not itself be polynomially recognizable but one may readily find a subclass of it that is polynomially recognizable and includes all the instances we need for NP-hardness.
In fact this will become clear when reading our proof but we refer also to~\cite{GPS12} where such a subclass is given for the case $(i,j)=(1,1)$.
A $(1,1)$-factor cut $(V_1,V_2)$ of $G$ is also known as a \emph{matching cut}, as no two edges in $E_G(V_1,V_2)$ have a common end-vertex, that is,
$E_G(V_1,V_2)$ is a {\it matching}.  Similarly {\sc $(1,1)$-Factor Cut with Roots} is known as {\sc Matching Cut with Roots} and was proved NP-complete by  Golovach, Paulusma and Song~\cite{GPS12} (by making an observation about the proof of the result of Patrignani and Pizzonia~\cite{PP01} that deciding whether or not any given graph has a matching cut is NP-complete).

We will prove the NP-completeness of {\sc $(i,j)$-Factor Cut with Roots} after first presenting a helpful lemma (a {\it clique} is a subset of vertices of $G$ that are pairwise adjacent to each other).

\begin{lemma} \label{lem:ijcliques}
Let $i$, $j$ and $k$ be positive integers where $i \leq j$ and $k>i+j$. Let $G$ be a graph that contains a clique $K$ on $k$ vertices.  Then, for every $(i,j)$-factor cut $(V_1,V_2)$ of $G$, either $V_K \subseteq V_1$ or $V_K \subseteq V_2$.
\end{lemma}  
\begin{proof}
If the lemma is false, then for some $(i,j)$-factor cut $(V_1,V_2)$, we can choose $v_1 \in V_1 \cap V_K$ and $v_2 \in V_2 \cap V_K$.  Let $M=E_G(V_1,V_2)$. Since every vertex in $V_1 \cap V_K$ is linked by an edge of $M$ to $v_2$ and every vertex in $V_2 \cap V_K$ is linked by an edge of $M$ to $v_1$, we have $d_M(v_1)+d_M(v_2) \geq k >i+j$, contradicting the definition of an $(i,j)$-factor cut.
\end{proof}
  
\begin{theorem}
\label{shc:the_fc_theorem}
Let $i$ and $j$ be positive integers, $i\leq j$. Then {\sc $(i,j)$-Factor Cut with Roots} is \NP-complete.
\end{theorem}

\begin{proof}
If $i=j=1$, then the problem is {\sc Matching Cut with Roots} which, as we noted, is known to be NP-complete~\cite{GPS12}.  We split the remaining cases in two according to whether or not $i=1$.  In each case, we construct a polynomial time reduction from {\sc Matching Cut with Roots}.  In particular, we take an instance $(G,s,t)$ of {\sc Matching Cut with Roots}, and construct a graph $G'$ that is a supergraph of $G=(V,E)$ and show that 
\begin{enumerate}[(1)]
\item \label{cond:1} $(G',s,t)$ is an instance of {\sc $(i,j)$-Factor Cut with Roots} (that is, if $G'$ has an $(i,j)$-factor cut $(V'_1,V'_2)$, then $s \in V_1$ and $t \in V_2$ or, possibly, vice versa if $i=j$), 
\item \label{cond:3} if $G'$ has an $(i,j)$-factor cut, then $G$ has a matching cut, and
\item \label{cond:2} if $G$ has a matching cut, then $G'$ has an $(i,j)$-factor cut.
\end{enumerate}
We note that (\ref{cond:1}) is an atypical feature of an NP-completeness proof as, unusually for {\sc $(i,j)$-Factor Cut with Roots}, it is not immediate to recognize a problem instance. We let $n=|V|$.

\medskip
\noindent \textbf{Case 1}: $i=1$.\\
\noindent 
Let $k = \max\{(n-1)(j-1), 1+j\}$.
Construct $G'$ from $G$ by first adding a complete graph $K$ on $k$ vertices and adding edges from $s$ to every vertex of $V_K$.  Then, for each $v \in V_G \setminus \{s\}$, add edges from $v$ to $j-1$ vertices of $K$ in such a way that no vertex of $V_K$ has more than one neighbour in $V_G \setminus \{s\}$.  

Let $(V'_1,V'_2)$ be a $(1,j)$-factor cut of $G'$.
The vertices of $\{s\} \cup V_K$ induce a clique on $1+k > 1+j$ vertices.  So, by Lemma~\ref{lem:ijcliques}, $\{s\} \cup V_K \subseteq V'_1$ or $\{s\} \cup V_K \subseteq V'_2$.  

Suppose that $\{s\} \cup V_K \subseteq V'_2$.  Then $V_G$ must contain vertices of both $V'_1$ (otherwise $V'_1$ would be empty) and $V'_2$ (at least $s$). Thus, as $G$ is connected, we can find a vertex $v \in V'_1 \cap V_G$ that has a neighbour in $V'_2 \cap V_G$.  But $v$ also has $j-1 \geq 1$ neighbours in $V_K$ and so has at least 2 neighbours in $V'_2$, contradicting the definition of a $(1,j)$-factor cut.

So we must have that $\{s\} \cup V_K \subseteq V'_1$.  Let $V_1 = V'_1 \cap V_G$ and $V_2=V'_2$ be a partition of $V_G$, and let $M=E_G(V_1,V_2)$ and $M'=E_G(V'_1,V'_2)$ and notice that $M'$ is the union of $M$ and, for each $v \in V_2$, the $j-1$ edges from $v$ to $V_K$.  For each $v \in V_1$, $d_M(v) = d_{M'}(v) \leq 1$.   For each $v \in V_2$, $d_M(v) = d_{M'}(v) - (j-1) \leq 1$.  So $(V_1,V_2)$ is a matching cut of $G$; this proves (\ref{cond:3}). And as $s \in V_1$, we have, by the definition of factor roots, $t \in V_2$; this proves (\ref{cond:1}).

To prove (\ref{cond:2}), we note that if $(V_1,V_2)$ is a matching cut of $G$, then we can assume that $s \in V_1$ and $t \in V_2$ (else relabel them for the purpose of constructing~$G'$), and then $(V_1 \cup V_K,V_2)$ is a $(1,j)$-factor cut of~$G'$.

\medskip
\noindent \textbf{Case 2}: $i\geq 2$.\\ 
\noindent
Let $k = \max\{(n-1)(j-1), i+j\}$.
Construct $G'$ from $G$ by first adding a complete graph $K^s$ on $k$ vertices and adding edges from $s$ to every vertex of $V_{K^s}$, and then adding a complete graph $K^t$ on $k$ vertices and adding edges from $t$ to every vertex of $V_{K^t}$.  Then, for each $v \in V_G \setminus \{s\}$, add edges from $v$ to $j-1$ vertices of $K^s$ in such a way that no vertex of $V_{K^s}$ has more than one neighbour in $V_G \setminus \{s\}$.   Afterwards, for each $v \in V_G \setminus \{t\}$, add edges from $v$ to $i-1$ vertices of $K^t$ in such a way that no vertex of $V_{K^t}$ has more than one neighbour in $V_G \setminus \{t\}$.  

Let $(V'_1,V'_2)$ be an $(i,j)$-factor cut of $G'$.  The vertices of $\{s\} \cup V_{K^s}$ induce a clique on at least $1+k > i+j$ vertices.  So, by Lemma~\ref{lem:ijcliques}, $\{s\} \cup V_{K^s} \subseteq V'_1$ or $\{s\} \cup V_{K^s} \subseteq V'_2$.  Similarly $\{t\} \cup V_{K^t} \subseteq V'_1$ or $\{t\} \cup V_{K^t} \subseteq V'_2$.

Suppose that $\{s\} \cup V_{K^s}$ and $\{t\} \cup V_{K^t}$ are both subsets of $V'_1$.  Then $V_G$ must contain vertices of both $V'_1$ (at least $s$ and $t$) and $V'_2$ (else it would be empty). Thus, as $G$ is connected, we can find a vertex $v \in V'_2 \cap V_G$ that has a neighbour in $V'_1 \cap V_G$.  But $v$ also has $j-1$ neighbours in $V_{K^s}$ and $i-1$ neighbours in $V_{K^t}$ and so has at least $1+(i-1)+(j-1)=i+j-1>j\geq i$ neighbours in $V'_2$, contradicting the definition of an $(i,j)$-factor. By an analogous argument $\{s\} \cup V_{K^s}$ and $\{t\} \cup V_{K^t}$ cannot both be subsets of $V'_2$.

Suppose that $i<j$ and $\{s\} \cup V_{K^s} \subseteq V'_2$.  As $G$ is connected and $V_G$ contains vertices of both $V'_1$ and $V'_2$, we can find a vertex $v \in V'_1 \cap V_G$ that has a neighbour in $V'_2 \cap V_G$.  But $v$ also has $j-1 > i-1$ neighbours in $V_{K^s}$ and so has more than $i$ neighbours in $V'_2$, contradicting the definition of a $(i,j)$-factor.

Thus we have that $\{s\} \cup V_{K^s}$ and $\{t\} \cup V_{K^t}$ are subsets of separate parts and, moreover, either $\{s\} \cup V_{K^s} \subseteq V'_1$ or $i=j$. Thus (\ref{cond:1}) is proved, and we have, in either case, that each vertex in $V'_1 \cap V_G$ is joined by $i-1$ edges to vertices in $V'_2 \setminus V_G$, and each vertex in $V'_2 \cap V_G$ is joined by $j-1$ edges to vertices in $V'_1 \setminus V_G$.
Therefore each vertex in $V'_1 \cap V_G$ is joined to at most one vertex in $V'_2 \cap V_G$, and each vertex in $V'_2 \cap V_G$ is joined to at most one vertex in $V'_1 \cap V_G$.  Thus $(V'_1\cap V_G, V'_2\cap V_G)$ is a matching cut of $G$.
This proves (\ref{cond:3}). 

To prove (\ref{cond:2}), we note that if $(V_1,V_2)$ is a matching cut of $G$, then we can assume that $s \in V_1$ and $t \in V_2$ (else relabel them for the purpose of constructing~$G'$), and then $(V_1 \cup V_{K^s},V_2 \cup V_{K^t})$ is an $(i,j)$-factor cut of~$G'$.
\end{proof}

\subsection{The Hardness Reduction}\label{s-hard}

Let $H$ be a connected 2-reflexive target graph. Let $p$ and $q$ be the two (non-adjacent) reflexive 
vertices of $H$. 
The {\it length} of a path is its number of edges.
The {\it distance} between two vertices $u$ and $v$ in a graph $G$ is the length of a shortest path between them and is denoted $\dist_G(u,v)$. 
We define two induced subgraphs $H_1$ and $H_2$ of $H$ whose vertex sets partition $V_H$. First $H_1$ contains those vertices of $H$ that are closer to $p$ than to $q$; and $H_2$ contains those vertices that are at least as close to $q$ as to $p$ (so contains any vertex equidistant to $p$ and $q$).  That is, $V_{H_1}=\left\{ v \in V_H  : \dist_H (v, p) < \dist_H (v, q) \right\}$ and $V_{H_2} = \left\{ v \in V_H  : \dist_H (v, q) \le \dist_H (v, p) \right\}$.
See Figure~\ref{shc:fig:H} for an example.
The following lemma follows immediately from our assumption that $H$ is connected.

\begin{lemma}
\label{shc:claim_cc}
Both $H_1$ and $H_2$ are connected.
Moreover, $\dist_{H_1}(x,p) = \dist_H (x, p)$ for every $x\in V_{H_1}$ and $\dist_{H_2}(x,q) = \dist_H (x,q)$ for every $x\in V_{H_2}$.
\end{lemma}

Let $\omega$ denote the size of a largest clique in $H$.
From graphs $H_1$ and $H_2$ we construct graphs $F_1$ and $F_2$, respectively, in the following way:

\begin{enumerate}
	\item for each $x\notin \{p,q\}$, create a vertex $t^1_x$;
	\item for $p$, create $\omega$ vertices $t^1_p,\ldots,t^{\omega}_p$;
	\item for $q$, create $\omega$ vertices $t^1_q,\ldots,t^{\omega}_q$;
	\item for $i=1,2$, add an edge in $F_i$ between any two vertices $t^h_x$ and $t^j_y$ if and only if $xy$ is an edge of $E_{H_i}$.
\end{enumerate}
We note that $F_1$ is the graph obtained by taking $H_1$ and replacing $p$ by a clique of size $\omega$.
Similarly, $F_2$ is the graph obtained by taking $H_2$ and replacing $q$ by a clique of size $\omega$.
We say that $t^1_p, \ldots, t^{\omega}_p$ are the {\it roots} of $F_1$ and that $t^1_q, \ldots, t^{\omega}_q$ are the {\it roots} of $F_2$. 
Figure~\ref{shc:fig:eg-graph-H} shows an example of the graphs $F_1$ and $F_2$ obtained from the graph $H$ in Figure~\ref{shc:fig:H}.

\begin{figure}[ht]

	\centering
	\begin{tikzpicture}[scale=0.85]

		\coordinate (a)  at ( 0.0, 0);
		\coordinate (b)  at ( 1.5, 1);
		\coordinate (c)  at ( 1.5,-1);
		\coordinate (f)  at ( 3.0, 0);
		\coordinate (h)  at ( 4.5, 0);
		
		\fill (a) circle[radius=3pt];
		\fill (b) circle[radius=3pt];
		\fill (c) circle[radius=3pt];
		\fill (f) circle[radius=3pt];
		\fill (h) circle[radius=3pt];

		\draw (a) -- (b);
		\draw (a) -- (c);
		
		\draw (b) -- (c);
		
		\draw (b) -- (f);
		
		\draw (c) -- (f);
		
		\draw (f) -- (h);
		
		\path[-] (a) edge  [in=60,out=120,loop] node {} ();
		\path[-] (h) edge  [in=60,out=120,loop] node {} ();

		\node[below = 0.15cm of a] (peeee) {$p$};
		\node[below = 0.15cm of h] (queue) {$q$};

		\coordinate (s)  at ( -1.5,  1);
		\coordinate (t)  at ( -1.5, -1);
		\coordinate (u)  at ( -3.0,  1);
		\coordinate (v)  at ( -3.0, -1);
		
		\coordinate (w)  at (  6.0,  1);
		\coordinate (x)  at (  6.0, -1);
		\coordinate (y)  at (  7.5,  0);
		
		\fill (s) circle[radius=3pt];
		\fill (t) circle[radius=3pt];
		\fill (u) circle[radius=3pt];
		\fill (v) circle[radius=3pt];
		\fill (w) circle[radius=3pt];
		\fill (x) circle[radius=3pt];
		\fill (y) circle[radius=3pt];
		
		\draw (a) -- (s);
		\draw (a) -- (t);
		
		\draw (s) -- (u);
		\draw (t) -- (v);
		
		\draw (h) -- (w);
		\draw (h) -- (x);
		\draw (w) -- (x);
		
		\draw (y) -- (w);
		\draw (y) -- (x);

		\draw[dashed] \convexpath{b,c,v,u}{0.5cm};
		\draw[dashed] \convexpath{f,w,y,x}{0.5cm};
		
		\node[below right = 0.5cm of c] (z) {$H_1$};
		\node[below = 0.8cm of y] (z) {$H_2$};

	\end{tikzpicture}

	\caption{An example of the construction of graphs $H_1$ and $H_2$ from a connected 2-reflexive target graph~$H$ with $\omega=3$.}
	\label{shc:fig:H}
\end{figure}
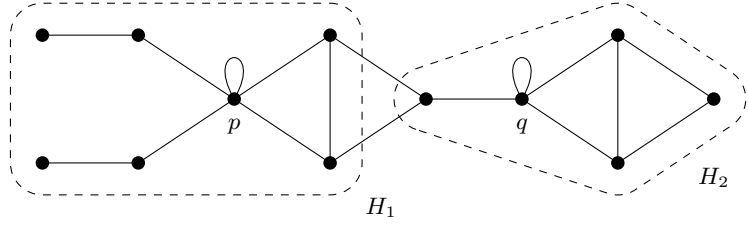

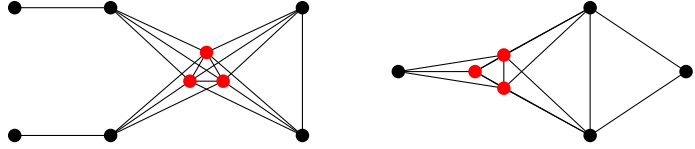
\begin{figure}[ht]

	\centering
	\begin{tikzpicture}[scale=0.85]

		\coordinate (a)  at ( 0.0, 0);
		\coordinate (a1) at ($ (a) + ( 90:0.3) $);
		\coordinate (a2) at ($ (a) + (210:0.3) $);
		\coordinate (a3) at ($ (a) + (330:0.3) $);
		\coordinate (b)  at ( 1.5, 1);
		\coordinate (c)  at ( 1.5,-1);
		\coordinate (f)  at ( 3.0, 0);
		\coordinate (h)  at ( 4.5, 0);
		\coordinate (h1) at ($ (h) + (180:0.3) $);
		\coordinate (h2) at ($ (h) + ( 60:0.3) $);
		\coordinate (h3) at ($ (h) + (-60:0.3) $);

		\draw (a1) -- (a2);
		\draw (a2) -- (a3);
		\draw (a3) -- (a1);
		
		\draw (h1) -- (h2);
		\draw (h2) -- (h3);
		\draw (h3) -- (h1);
		
		\draw (a1) -- (b);
		\draw (a1) -- (c);
		\draw (a2) -- (b);
		\draw (a2) -- (c);
		\draw (a3) -- (b);
		\draw (a3) -- (c);
		
		\draw (b) -- (c);
		
		\draw (f) -- (h1);
		\draw (f) -- (h2);
		\draw (f) -- (h3);	
		
		\coordinate (s)  at ( -1.5,  1);
		\coordinate (t)  at ( -1.5, -1);
		\coordinate (u)  at ( -3.0,  1);
		\coordinate (v)  at ( -3.0, -1);
		
		\coordinate (w)  at (  6.0,  1);
		\coordinate (x)  at (  6.0, -1);
		\coordinate (y)  at (  7.5,  0);
		
		\fill (s) circle[radius=3pt];
		\fill (t) circle[radius=3pt];
		\fill (u) circle[radius=3pt];
		\fill (v) circle[radius=3pt];
		\fill (w) circle[radius=3pt];
		\fill (x) circle[radius=3pt];
		\fill (y) circle[radius=3pt];
		
		\draw (a1) -- (s);
		\draw (a1) -- (t);
		\draw (a2) -- (s);
		\draw (a2) -- (t);
		\draw (a3) -- (s);
		\draw (a3) -- (t);
		
		\draw (s) -- (u);
		\draw (t) -- (v);
		
		\draw (h1) -- (w);
		\draw (h1) -- (x);
		\draw (h2) -- (w);
		\draw (h2) -- (x);
		\draw (h3) -- (w);
		\draw (h3) -- (x);
		\draw (w) -- (x);
		
		\draw (y) -- (w);
		\draw (y) -- (x);
		
		\fill[red] (a1) circle[radius=3pt];
		\fill[red] (a2) circle[radius=3pt];
		\fill[red] (a3) circle[radius=3pt];
		\fill (b) circle[radius=3pt];
		\fill (c) circle[radius=3pt];
		\fill (f) circle[radius=3pt];
		\fill[red] (h1) circle[radius=3pt];
		\fill[red] (h2) circle[radius=3pt];
		\fill[red] (h3) circle[radius=3pt];

	\end{tikzpicture}

	\caption{The graphs $F_1$ (left) and $F_2$ (right) resulting from the graph $H$ in Figure~\ref{shc:fig:H}.}
	\label{shc:fig:eg-graph-H}
\end{figure}

Let $\ell=\dist_H(p,q) \geq 2$ denote the distance between $p$ and $q$.
Let $N_p$ be the set of neighbours of $p$ 
that are each on some shortest path (thus of length~$\ell$) from $p$ to $q$ in $H$. Let $r_p$ be the size of a largest clique in $N_p$. We define 
$N_q$ and $r_q$ similarly.
We will reduce from $(r_p,r_q)$-{\sc Factor Cut with Roots}, which is \NP-complete due to
Theorem~\ref{shc:the_fc_theorem}.
Hence, consider an instance $(G,s,t)$ of $(r_p,r_q)$-{\sc Factor Cut with Roots}, where $G$ is a connected graph  
and
$s$ and $t$ form the (ordered) pair of
$(r_p,r_q)$-factor roots of $G$. 
Recall that we assume that $G$ is irreflexive.

We say that we {\em identify} two vertices $u$ and $v$ of a graph when we remove them from the graph and replace them with a single vertex that we make adjacent to every vertex that was adjacent to $u$ or $v$.
From $F_1$, $F_2$, and $G$ we construct a new graph $G^\prime$ as follows:

\begin{enumerate}	
	\item For each edge $e = uv \in E_G$, we do as follows. We create four vertices, $g_{u,e}^{\mathrm{r}}$, $g_{u,e}^{\mathrm{b}}$, $g_{v,e}^{\mathrm{r}}$ and $g_{v,e}^{\mathrm{b}}$. We also create two paths $P_e^1$ and $P_e^2$, each of length $\ell-2$, between $g_{u,e}^{\mathrm{r}}$ and $g_{v,e}^{\mathrm{b}}$, and between $g_{v,e}^{\mathrm{r}}$ and $g_{u,e}^{\mathrm{b}}$, respectively. If $\ell=2$ we identify $g_{u,e}^{\mathrm{r}}$ and $g_{v,e}^{\mathrm{b}}$ and $g_{v,e}^{\mathrm{r}}$ and $g_{u,e}^{\mathrm{b}}$ to get paths of length~0.
	\item For each vertex $u \in V_G$, we do as follows. 
	First we construct a clique $C_u$ on $\omega$ vertices. We denote these vertices by $g^1_u, \ldots, g^{\omega}_u$. We then make every vertex in $C_u$ adjacent to both $g_{u,e}^{\mathrm{r}}$ and $g_{u,e}^{\mathrm{b}}$ for every edge $e$ incident to $u$; we call $g_{u,e}^{\mathrm{r}}$ and $g_{u,e}^{\mathrm{b}}$ a {\it red} and {\it blue} neighbour of $C_u$, respectively; if $\ell=2$, then the vertex obtained by identifying two vertices
	 $g_{u,e}^{\mathrm{r}}$ and $g_{v,e}^{\mathrm{b}}$, or $g_{v,e}^{\mathrm{r}}$ and $g_{u,e}^{\mathrm{b}}$ is simultaneously a red neighbour of one clique and a blue neighbour of another one.
Finally, for every two edges $e$ and $e^\prime$ incident to $u$, we make $g_{u,e}^{\mathrm{r}}$ and $g_{u,e^\prime}^{\mathrm{r}}$ adjacent, that is, the set of red neighbours of $C_u$ form a clique, whereas the set of blue neighbours form an independent set.
	\item We add $F_1$ by identifying $t^i_p$ and $g^i_s$ for $i = 1,\ldots,\omega$, and we add $F_2$ by identifying $t^i_q$ and $g^i_t$ for $i = 1,\ldots,\omega$. We denote the vertices in $F_1$ and $F_2$ in $G^\prime$ by their label $t^i_x$ in $F_1$ or $F_2$.
\end{enumerate}
See Figure~\ref{shc:fig:the-graph-GGprs} for an example of a graph $G'$. 
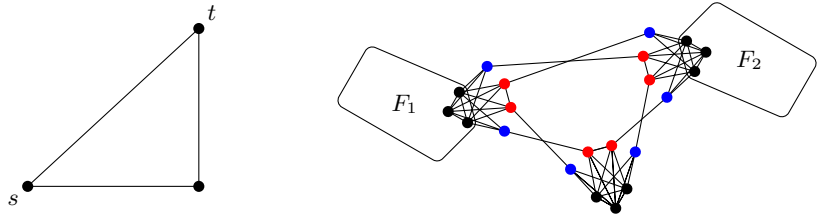
\begin{figure}[!ht]
    \centering
    \begin{subfigure}[t]{0.43\textwidth}

		\centering

		\begin{tikzpicture}[scale=0.7]

			\coordinate (s)  at (0.75, 0   );
			\coordinate (d)  at (4   , 0);
			\coordinate (t)  at (4   , 3);

			\fill (s) circle[radius=3pt];
			\fill (t) circle[radius=3pt];
			\fill (d) circle[radius=3pt];

			\draw (d) -- (s) -- (t) -- (d);

			\node[below left  = 0cm of s] (y) {$s$};
			\node[above right = 0cm of t] (z) {$t$};

		\end{tikzpicture}
		\caption{An example of a graph $G$ with a $(1,2)$-factor cut with	
$(1,2)$-factor roots $s$ and $t$.}
	\end{subfigure}
	~
    \begin{subfigure}[t]{0.53\textwidth}

		\centering
		\begin{tikzpicture}[scale=0.7, rotate=-30]

			\coordinate (s)  at (0.75, 0   );
			\coordinate (d)  at (4   , 0);
			\coordinate (t)  at (4   , 3);
			
			\coordinate (s1) at ($ (s) + (135:0.3) $);
			\coordinate (s2) at ($ (s) + (225:0.3) $);
			\coordinate (s3) at ($ (s) + (-45:0.3) $);
			
			\coordinate (s4) at ($ (s) + (  0:0.9) $);
			\coordinate (s5) at ($ (s) + ( 30:0.9) $);
			\coordinate (s6) at ($ (s) + ( 60:0.9) $);
			\coordinate (s7) at ($ (s) + ( 90:0.9) $);

			\coordinate (t1) at ($ (t) + ( 45:0.3) $);
			\coordinate (t2) at ($ (t) + (135:0.3) $);
			\coordinate (t3) at ($ (t) + (-45:0.3) $);
			
			\coordinate (t4) at ($ (t) + (180:0.9) $);
			\coordinate (t5) at ($ (t) + (210:0.9) $);
			\coordinate (t6) at ($ (t) + (240:0.9) $);
			\coordinate (t7) at ($ (t) + (270:0.9) $);
			
			\coordinate (d0) at ($ (d) + ( 45:0.3) $);
			\coordinate (d1) at ($ (d) + (225:0.3) $);
			\coordinate (dy) at ($ (d) + (-45:0.3) $);
			\coordinate (dx) at (dy);
			
			\coordinate (d2) at ($ (d) + ( 90:0.9) $);
			\coordinate (d3) at ($ (d) + (120:0.9) $);
			\coordinate (d4) at ($ (d) + (150:0.9) $);
			\coordinate (d5) at ($ (d) + (180:0.9) $);

			\coordinate (T13) at ($ (t) + (2,0) $);
			\coordinate (T12) at ($ (t) + (2,1) $);
			\coordinate (T11) at ($ (t) + (0,1) $);
			
			\coordinate (S13) at ($ (s) + (-2, 0) $);
			\coordinate (S12) at ($ (s) + (-2,-1) $);
			\coordinate (S11) at ($ (s) + ( 0,-1) $);
			
			\coordinate (TL) at ($ (t) + (1,0.5) $);
			\coordinate (SL) at ($ (s) + (-1,-0.5) $);

			\draw (s2) -- (s4) -- (s3) -- (s5) -- (s2) -- (s6) -- (s3) -- (s7) -- (s2);
			\draw (s4) -- (s1) -- (s5) -- (s6) -- (s1) -- (s7);
			
			\draw (t2) -- (t4) -- (t3) -- (t5) -- (t2) -- (t6) -- (t3) -- (t7) -- (t2);
			\draw (t4) -- (t1) -- (t5) -- (t6) -- (t1) -- (t7);
			
			\draw (s1) -- (s2) -- (s3) -- (s1);
			\draw (t1) -- (t2) -- (t3) -- (t1);
			\draw (d2) -- (d1) -- (d3) -- (d4) -- (d1) -- (d5);
			\draw (d2) -- (d0) -- (d3) -- (d4) -- (d0) -- (d5);
			\draw (d1) -- (d0);
			
			\draw (d3) -- (t7);
			\draw (d2) -- (t6);
			
			\draw (s7) -- (t5);
			\draw (s6) -- (t4);
			
			\draw (s5) -- (d5);
			\draw (s4) -- (d4);
			
			\draw (dx) -- (d0) -- (dy) -- (d1) -- (dx) -- (dy);
			\draw (d2) -- (dx) -- (d3) -- (dy) -- (d2);
			\draw (d4) -- (dx) -- (d5) -- (dy) -- (d4);
			
			\fill (s1) circle[radius=3pt];
			\fill (s2) circle[radius=3pt];
			\fill (s3) circle[radius=3pt];
			\fill[blue] (s4) circle[radius=3pt];
			\fill[red] (s5) circle[radius=3pt];
			\fill[red] (s6) circle[radius=3pt];
			\fill[blue] (s7) circle[radius=3pt];
			
			\fill (t1) circle[radius=3pt];
			\fill (t2) circle[radius=3pt];
			\fill (t3) circle[radius=3pt];
			\fill[blue] (t4) circle[radius=3pt];
			\fill[red] (t5) circle[radius=3pt];
			\fill[red] (t6) circle[radius=3pt];
			\fill[blue] (t7) circle[radius=3pt];
			
			\fill (d0) circle[radius=3pt];
			\fill (d1) circle[radius=3pt];
			\fill (dy) circle[radius=3pt];
			\fill[blue] (d2) circle[radius=3pt];
			\fill[red] (d3) circle[radius=3pt];
			\fill[red] (d4) circle[radius=3pt];
			\fill[blue] (d5) circle[radius=3pt];

			\draw \convexpath{t3, t2, T11, T12, T13}{0.15cm};
			\draw \convexpath{s1, s3, S11, S12, S13}{0.15cm};

			\node at (TL) {$F_2$};
			\node at (SL) {$F_1$};
			
		\end{tikzpicture}
		\caption{The corresponding graph $G^\prime$ where $H$ is a 2-reflexive target graph with $\ell = 3$ and $\omega = 3$.}
	\end{subfigure}
	
	\caption{An example of a graph $G$ and the corresponding graph $G^\prime$.}
	\label{shc:fig:the-graph-GGprs}
\end{figure}
The next lemma describes a straightforward property of graph homomorphisms that will prove useful.
\begin{lemma}
\label{shc:lemma:shortcuts}
If there exists a homomorphism $h \, : \, G^\prime \rightarrow H$ then $\dist_{G^\prime}(u, v) \ge \dist_H \left( h(u), h(v) \right)$ for every pair of vertices $u,v \in V_{G^\prime}$.
\end{lemma}

We now prove the key property of our construction.

\begin{lemma}
\label{shc:claim1}
For every homomorphism $h$ from $G^\prime$ to $H$, there exists at least one clique $C_a$ with $p\in h(C_a)$ and at least one clique $C_b$ with
$q\in h(C_b)$.
\end{lemma}

\begin{proof}
Since for each $u \in V_G$ and any edge $e$ incident to $u$, every clique $C_u \cup \{g_{u,e}^{\mathrm{r}}\}$ in $G'$ is of size at least $\omega+1$, we find that $h$ must map at least two of its vertices to a reflexive vertex, so either to $p$ or $q$. Hence, for every $u\in V_G$, we find that $h$ maps at least
one vertex of $C_u$ to either $p$ or $q$.
  
We prove the lemma by contradiction.  We will assume that $h$ does not map any vertex of any $C_u$ to $q$, thus $p \in h(C_u)$ for all $u\in V_G$.
We will note later that if instead $q\in h(C_u)$ for all $u\in V_G$ we can obtain a contradiction in the same way.

We consider two vertices $t^i_p \in F_1$ and $t^j_q \in F_2$ such that $h(t^i_p) = h(t^j_q) = p$.
Without loss of generality let $i=j=1$.
We shall refer to these vertices as $t_p$ and $t_q$ respectively.
We now consider a vertex $v \in V_{F_1} \disjun V_{F_2}$. By Lemma \ref{shc:lemma:shortcuts}, $\dist_{G^\prime}(v,t_p) \ge \dist_H(h(v),p)$ and $\dist_{G^\prime}(v,t_q) \ge \dist_H(h(v),p)$. In other words:
\begin{equation*}
\min \left(\dist_{G^\prime}(v,t_p), \dist_{G^\prime}(v,t_q) \right) \ge \dist_H(h(v),p).
\end{equation*}
In fact by applying Lemma~\ref{shc:lemma:shortcuts} we can generalize this further to any vertex mapped to $p$ by $h$:
\begin{equation}
\label{shc:eqn:anyw}
\min_{w \in h^{-1}(p)} \left(\dist_{G^\prime}(v,w) \right) \ge \dist_H(h(v),p).
\end{equation}
For every $v \in V_{G^\prime}$ we define a value $\dee(v)$ as follows:

\begin{equation*}
\dee(v) = \left\{
  \begin{array}{ll}
    \dist_{F_1} (v, t_p) \quad & \mathrm{if~} v \in F_1 \\
    \dist_{F_2} (v, t_q) & \mathrm{if~} v \in F_2 \\
    \lfloor \ell / 2 \rfloor & \mathrm{otherwise} \\
  \end{array}
\right.
\end{equation*}

\begin{nclaim}
\label{shc:lemma:d_of_v}
$\dee(v) \ge \min_{w \in h^{-1}(p)} \left(\dist_{G^\prime}(v,w) \right) \ge \dist_H(h(v),p)$ for all $v \in V_{G^\prime}$.
\end{nclaim}

\noindent
We prove Claim \ref{shc:lemma:d_of_v} by showing that 
$\dee(v) \ge \min_{w \in h^{-1}(p)} \left(\dist_{G^\prime}(v,w) \right)$, which suffices due to~(\ref{shc:eqn:anyw}).
First suppose $v \in V_{F_1} \cup V_{F_2}$. We may assume, without loss of generality, that $v \in V_{F_2}$.
So $\dee(v)=\dist_{F_2}(v, t_q)=\dist_{G^\prime}(v, t_q)\ge \min_{w \in h^{-1}(p)} \left(\dist_{G^\prime}(v,w) \right)$, as $t_q \in h^{-1}(p)$.

Now suppose $v \not\in V_{F_1} \cup V_{F_2}$. Then $v$ either belongs to a clique $C_u$ or is a vertex of a path $P^1_e$ or $P^2_e$ between two cliques.
If $v$ belongs to a clique or is an end-vertex of such a path, then $v$ is
either in $h^{-1}(p)$ or adjacent to a vertex in $h^{-1}(p)$ (since at least one vertex in $C_u$ maps to $p$). Hence $\dee(v)= \lfloor \ell / 2 \rfloor\geq 1\geq\min_{w \in h^{-1}(p)} \left(\dist_{G^\prime}(v,w) \right)$.
Finally, suppose $v$ is an inner vertex of a path $P^1_e$ or $P^2_e$. By definition, such a path has length~$\ell-2$. 
Then $v$ is at most distance $\lfloor (\ell - 2) / 2 \rfloor$ from a vertex in a clique, which we know is either in $h^{-1}(p)$ or adjacent to a vertex in $h^{-1}(p)$. Hence $\dee(v)=    \lfloor \ell / 2 \rfloor =  \lfloor (\ell - 2) / 2 \rfloor + 1 \geq \min_{w \in h^{-1}(p)} \left(\dist_{G^\prime}(v,w) \right)$.
This proves Claim \ref{shc:lemma:d_of_v}.

\begin{nclaim}
\label{shc:required_reach}
If there exists a surjective homomorphism from $G^\prime$ to $H$, then for any integer $d \ge \ell$:
\begin{equation*}
\left| \left\{ t^1_w \in V_{F_1} \disjun V_{F_2}  : \dee(t^1_w) \ge d \right\} \right| \ge \left| \left\{ w \in V_H  : \dist_H(w, p) \ge d \right\} \right|.
\end{equation*}
\end{nclaim}

\noindent
We prove Claim~\ref{shc:required_reach} as follows.
Using the fact that with a surjective homomorphism every vertex must be mapped to, we see from Lemma \ref{shc:lemma:shortcuts} that if there are $n$ vertices in $H$ which are at a distance $d$ from $p$, there must be at least $n$ vertices in $G^\prime$ that are at distance at least $d$ from every vertex that maps to $p$. This means we can say for any distance $d \ge 0$:

\begin{equation*}
|\{ v \in V_{G^\prime}  : \min_{w \in h^{-1}(p)} \left(\dist_{G^\prime}(v,w) \right) \ge d\}| \ge \left| \left\{ w \in V_H  : \dist_H(w, p) \ge d \right\} \right|.
\end{equation*}

Combining this inequality with Claim \ref{shc:lemma:d_of_v} yields, for every distance $d \ge 0$:

\begin{equation*}
\left| \left\{ v \in V_{G^\prime}  : \dee(v) \ge d \right\} \right| \ge \left| \left\{ w \in V_H  : \dist_H(w, p) \ge d \right\} \right|.
\end{equation*}
Now let $d\geq \ell$. Then we only have to consider vertices in $F_1 \disjun F_2$. Hence, for every $d \ge \ell$:
\begin{equation*}
\left| \left\{ t^i_w \in V_{F_1} \disjun V_{F_2}  : \dee(t^i_w) \ge d \right\} \right| \ge \left| \left\{ w \in V_H  : \dist_H(w, p) \ge d \right\} \right|.
\end{equation*}
By construction, for any $t^i_w$ with $i>1$ we have that $w\in \{s,t\}$ and thus $\dee(t^i_w) \le 1 < \ell\leq d$.
Therefore, no vertex $t^i_w$ with $i \ne 1$ is involved in the equation above, so we can write:
\begin{equation*}
\left| \left\{ t^1_w \in V_{F_1} \disjun V_{F_2}  : \dee(t^1_w) \geq d \right\} \right| \ge \left| \left\{ w \in V_H  : \dist_H(w, p) \ge d \right\} \right|.
\end{equation*}
Hence Claim~\ref{shc:required_reach} is proven.

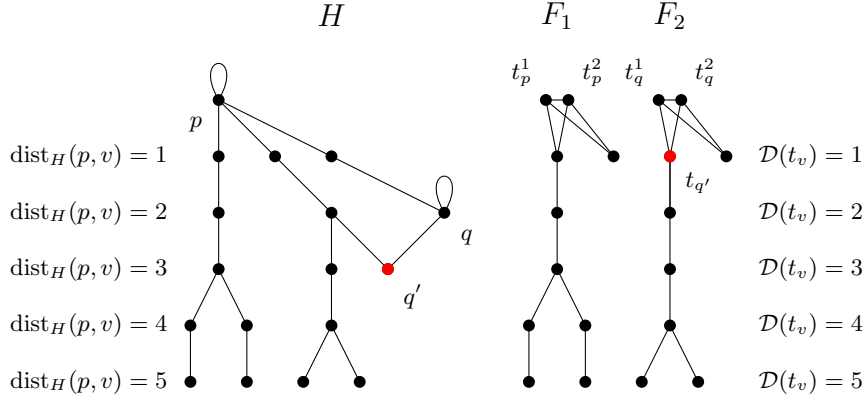
\begin{figure}[ht]
    \centering
	\begin{tikzpicture}[scale=0.75]			
		\begin{scope}[shift={(0, 0)}]
          		\coordinate (p)  at ( 0.0, -0.0);
			\coordinate (q)  at ( 4.0, -2.0);
		
			\coordinate (pp) at ( 2.0, -1.0);
		
			\coordinate (a)  at ( 1.0, -1.0);
			\coordinate (b)  at ( 2.0, -2.0);
			\coordinate (c)  at ( 3.0, -3.0);
		
			\coordinate (d)  at ( 2.0, -3.0);
			\coordinate (e)  at ( 2.0, -4.0);
			\coordinate (f)  at ( 1.5, -5.0);
			\coordinate (g)  at ( 2.5, -5.0);
		
			\coordinate (h)  at ( 0.0, -1.0);
			\coordinate (i)  at ( 0.0, -2.0);
			\coordinate (j)  at ( 0.0, -3.0);
			\coordinate (k)  at (-0.5, -4.0);
			\coordinate (l)  at (-0.5, -5.0);
			\coordinate (m)  at ( 0.5, -4.0);
			\coordinate (n)  at ( 0.5, -5.0);
		
			\fill (p) circle[radius=3pt];
			\fill (q) circle[radius=3pt];
			\fill (pp) circle[radius=3pt];
			\fill (a) circle[radius=3pt];
			\fill (b) circle[radius=3pt];
			\fill (c) circle[radius=3pt];
			\fill (d) circle[radius=3pt];
			\fill (e) circle[radius=3pt];
			\fill (f) circle[radius=3pt];
			\fill (g) circle[radius=3pt];
			\fill (h) circle[radius=3pt];
			\fill (i) circle[radius=3pt];
			\fill (j) circle[radius=3pt];
			\fill (k) circle[radius=3pt];
			\fill (l) circle[radius=3pt];
			\fill (m) circle[radius=3pt];
			\fill (n) circle[radius=3pt];
		
			\path[-] (p) edge  [in=60,out=120,loop] node {} ();
			\path[-] (q) edge  [in=60,out=120,loop] node {} ();
		
			\draw (p) -- (pp) -- (q);
			\draw (p) -- (a) -- (b) -- (c) -- (q);
		
			\draw (b) -- (d) -- (e) -- (f);
			\draw (e) -- (g);
		
			\draw (p) -- (h) -- (i) -- (j) -- (k) -- (l);
			\draw (j) -- (m) -- (n);
		
			\node[below left = 0.15cm of p] (peeee) {$p$};
			\node[below right = 0.15cm of q] (queue) {$q$};
			\node[below right = 0.13cm of c] (queuep) {$q^\prime$};
			\fill[red] (c) circle[radius=3pt];

			\node (D1) at (-2.3, -1) {$\dist_H(p,v) = 1$};
			\node (D2) at (-2.3, -2) {$\dist_H(p,v) = 2$};
			\node (D3) at (-2.3, -3) {$\dist_H(p,v) = 3$};
			\node (D4) at (-2.3, -4) {$\dist_H(p,v) = 4$};
			\node (D5) at (-2.3, -5) {$\dist_H(p,v) = 5$};
					
			\node (H) at (2, 1.5) {\large $H$};		
		\end{scope}
			
		\begin{scope}[shift={(6, 0)}]
			\coordinate (p)  at ( 0.2, -0.0);
			\coordinate (p2) at (-0.2, -0.0);
			\coordinate (q)  at ( 2.2, -0.0);
			\coordinate (q2) at ( 1.8, -0.0);
		
			\coordinate (pp) at ( 3.0, -1.0);
		
			\coordinate (a)  at ( 1.0, -1.0);
			\coordinate (b)  at ( 2.0, -2.0);
			\coordinate (c)  at ( 2.0, -1.0);
		
			\coordinate (d)  at ( 2.0, -3.0);
			\coordinate (e)  at ( 2.0, -4.0);
			\coordinate (f)  at ( 1.5, -5.0);
			\coordinate (g)  at ( 2.5, -5.0);
		
			\coordinate (h)  at ( 0.0, -1.0);
			\coordinate (i)  at ( 0.0, -2.0);
			\coordinate (j)  at ( 0.0, -3.0);
			\coordinate (k)  at (-0.5, -4.0);
			\coordinate (l)  at (-0.5, -5.0);
			\coordinate (m)  at ( 0.5, -4.0);
			\coordinate (n)  at ( 0.5, -5.0);
		
			\fill (p) circle[radius=3pt];
			\fill (q) circle[radius=3pt];
			\fill (p2) circle[radius=3pt];
			\fill (q2) circle[radius=3pt];
			\fill (pp) circle[radius=3pt];
			\fill (a) circle[radius=3pt];
			\fill (b) circle[radius=3pt];
			\fill (c) circle[radius=3pt];
			\fill (d) circle[radius=3pt];
			\fill (e) circle[radius=3pt];
			\fill (f) circle[radius=3pt];
			\fill (g) circle[radius=3pt];
			\fill (h) circle[radius=3pt];
			\fill (i) circle[radius=3pt];
			\fill (j) circle[radius=3pt];
			\fill (k) circle[radius=3pt];
			\fill (l) circle[radius=3pt];
			\fill (m) circle[radius=3pt];
			\fill (n) circle[radius=3pt];
		
			\draw (q) -- (q2);
			\draw (p) -- (p2);
			
			\draw (q) -- (pp);
			\draw (q2) -- (pp);
			\draw (p) -- (a);
			\draw (p2) -- (a);
			\draw (b) -- (c) -- (q);
			\draw (b) -- (c) -- (q2);
		
			\draw (b) -- (d) -- (e) -- (f);
			\draw (e) -- (g);
		
			\draw (p) -- (h) -- (i) -- (j) -- (k) -- (l);
			\draw (p2) -- (h);
			\draw (j) -- (m) -- (n);
		
			\node[above left = 0.1cm of p2] (peeee) {$t_p^1$};
			\node[above left = 0.1cm of q2] (queue) {$t_q^1$};
			\node[above right = 0.1cm of p] (peeee) {$t_p^2$};
			\node[above right = 0.1cm of q] (queue) {$t_q^2$};
			
			\node[below right = 0.13cm of c] (queuep) {$t_{q^\prime}$};
			\fill[red] (c) circle[radius=3pt];
		
			\node (D1) at (4.5, -1) {$\dee(t_v) = 1$};
			\node (D2) at (4.5, -2) {$\dee(t_v) = 2$};
			\node (D3) at (4.5, -3) {$\dee(t_v) = 3$};
			\node (D4) at (4.5, -4) {$\dee(t_v) = 4$};
			\node (D5) at (4.5, -5) {$\dee(t_v) = 5$};		
			
			\node (F1) at (0, 1.5) {\large $F_1$};
			\node (F2) at (2, 1.5) {\large $F_2$};		
		\end{scope}		
	\end{tikzpicture}
	\caption{An example of a graph $H$ with corresponding graphs $F_1$ and $F_2$. Vertices in~$H$ equidistant from $p$ are plotted at the same vertical position and likewise vertices $t_v \in F_1$ and $t_w \in F_2$ with $\dee(t_v) = \dee(t_w)$ are plotted at the same vertical position. The vertices $q^\prime \in H$ and corresponding $t_{q^\prime} \in F_2$ are highlighted.}
	\label{shc:fig:required-reach}
\end{figure}

\medskip
\noindent
We first present the intuition behind the final part of the proof.
Consider the graphs $F_1$, $F_2$ and $H$ in the example shown in Figure~\ref{shc:fig:required-reach}.
We recall that every vertex $v$ (other than $p$ or $q$) has a single corresponding vertex $t_v$ in $F_1$ or $F_2$.
We may naturally want to map the vertices of $F_1$ onto the vertices of $H_1$, which is possible by definition of $F_1$.
However, when we try to map the vertices of $F_2$ onto the vertices of $H_2$, with $h(t_q^i) = p$ (for some $i$), we will prove that there is at least one vertex~$q^\prime$ in $H_2$ which is further from $p$ in $H$ than it is from $q$ and that cannot be mapped to and thus violates the surjectivity constraint.
In Figure~\ref{shc:fig:required-reach} this vertex, which will play a special role in our proof, is shown in red.
In the example of this figure, $\ell=3$ and we observe 
 that there are ten vertices in $H$ 
 (including $q'$) 
 with $\dist_H(p,v) \ge 3$ but only nine vertices 
 (excluding $q'$)
 in $F_1 \cup F_2$ with $\dee(t_v) \ge 3$ which could be mapped to these vertices. This contradicts Claim~\ref{shc:required_reach}. 

We now formally prove
that our initial assumption that $p \in h(C_u)$ for all $u\in V_G$ contradicts Claim~\ref{shc:required_reach}.
For every vertex $x$ in $H_1$ there is a corresponding vertex~$t^1_x$ such that 
$\dee(t^1_x)=\dist_{F_1}(t^1_x, t_p)=\dist_{H_1}(x, p)$, where the latter equality follows from the construction of $F_1$.
From Lemma \ref{shc:claim_cc} we find that $\dist_{H_1}(x, p)=\dist_{H}(x, p)$ for every $x\in V_{H_1}$.
Hence $\dee(t^1_x)=\dist_{H}(x, p)$, 
and for all $d \ge 0$:

\begin{equation}
\label{e-extra}
\left| \left\{ t^1_x \in V_{F_1} : \dee(t^1_x) \ge d \right\} \right| = \left| \left\{ x \in V_{H_1} : \dist_H(x, p) \ge d \right\} \right| .
\end{equation}

Now let $x\in V_{H_2}$. 
Using the same arguments,
we see that $\dee(t^1_x) = \dist_{H}(x, q)$, and thus $\dee(t^1_x) = \dist_{H}(x, q) \le \dist_{H}(x, p)$ by definition.
Note that, had we instead supposed that it was $q$ to which everything mapped, we would instead have a strict inequality.  As it turns out, we only need the weaker inequality. 

We now look for a vertex $q^\prime$ in $H_2$, such that $q^\prime$ is as far from $p$ as possible, subject to the condition that $\dist_H(q^\prime, q) < \dist_H(q^\prime, p)$. Let $j=\dist_H(q^\prime, p)$. We see that for any vertex $x$ in $H_2$ such that
$\dist_H(x, p) > j$, it is the case that $\dist_H(x, q) = \dist_H(x, p)$.
Note that there may be no vertices with $\dist_H(x, q) = \dist_H(x, p)$ in which case $q^\prime$ is simply the farthest vertex from $p$ within $H_2$.
We also observe that $q'=q$ is possible.
So $j$ is well defined and, in fact, we have that $j\geq \ell$.

We now consider the mapping of vertices in $H_2$ at a distance $d \ge \ell$ from $p$.
We recall that $\dee(t^1_x) = \dist_{H}(x, q)$ for every $x$ in $H_2$ and that for a vertex~$x\in H_2$ of distance at least $j+1$ from $q$ in $H$, it holds that $\dist_{H}(x, q) = \dist_{H}(x, p)$.
Combining this with equation~(\ref{e-extra}) yields that:

\begin{equation}
\label{shc:eq:tada_1}
\left| \left\{ t^1_x \in V_{F_1} \disjun V_{F_2} : \dee(t^1_x) > j \right\} \right| = \left| \left\{ x \in V_H : \dist_H(x, p) > j \right\} \right|.
\end{equation}

However, for $d = j$ we find that, in addition to vertices in $H_2$ equidistant from $p$ and $q$, there is at least one vertex that is closer to $q$ than $p$, namely $q^\prime$, for which it holds that $\dee(t^1_{q^\prime})= \dist_{H}(q^\prime, q)< \dist_H({q^\prime}, p) = j$.
It therefore follows that there are fewer vertices $t^1_{x}$ with $\dee(t^1_{x}) = j$ than there are vertices $x$ with $\dist_H({x}, p) = j$ and hence we see that:

\begin{equation}
\label{shc:eq:tada_2}
\left| \left\{ t^1_x \in V_{F_1} \disjun V_{F_2} : \dee(t^1_x) = j \right\} \right| < \left| \left\{ x \in V_H : \dist_H(x, p) = j \right\} \right|.
\end{equation}

By combining equations (\ref{shc:eq:tada_1}) and (\ref{shc:eq:tada_2}), we see that:

\begin{equation*}
\left| \left\{ t^1_x \in V_{F_1} \disjun V_{F_2} : \dee(t^1_x) \ge j \right\} \right| < \left| \left\{ x \in V_H : \dist_H(x, p) \ge j \right\} \right|.
\end{equation*}
As $j\geq \ell$, this contradicts Claim \ref{shc:required_reach} and concludes the proof of Lemma~\ref{shc:claim1}.
\end{proof}

We are now ready to state our main result.

\begin{theorem}
\label{shc:the_theorem-gt2}
For every connected $2$-reflexive graph $H$, the \shc{} problem is \NP-complete.
\end{theorem}

\begin{proof}
Let $H$ be a connected $2$-reflexive graph with reflexive vertices $p$ and $q$ at distance $\ell\geq 2$ from each other. 
Let $\omega$ be the size of a largest clique in $H$.
We define the graphs $H_1$, $H_2$, $F_1$ and $F_2$, sets $N_p$ an $N_q$, and values $r_p$, $r_q$ as above.
Recall that the problem is readily seen to be in \NP\ and that we reduce from $(r_p,r_q)$-{\sc Factor Cut with Roots}.
From $F_1, F_2$ and an 
instance~$(G,s,t)$ 
of the latter problem we construct the graph $G^\prime$.
We claim that $G$ has an $(r_p,r_q)$-factor cut $(V_1,V_2)$ if and only if there exists a surjective homomorphism~$h$ from 
$G^\prime$ to~$H$.

First suppose that $G$  has an $(r_p,r_q)$-factor cut $(V_1,V_2)$. By definition, $s\in V_1$ and $t\in V_2$.
We define a homomorphism $h$ as follows. For every $x\in V_{F_1}\cup V_{F_2}$, we let $h$ map
$t^1_x$ to $x$. This shows that $h$ is surjective. It remains to define $h$ on the other vertices.
For every $u\in V_G$, let $h$ map all of $C_u$ to $p$ if $u$ is in $V_1$ and let $h$ map all of $C_u$ to $q$ if $u$ is in $V_2$ (note that this is consistent with how we defined $h$ so far). For each $uv\in E_G$ with $u,v\in V_1$, we map the vertices of the paths $P^1_e$ and $P^2_e$ to $p$.
For each $uv\in E_G$ with $u,v\in V_2$, we map the vertices of the paths $P^1_e$ and $P^2_e$ to $q$. We are left to show that the vertices of the remaining paths $P^1_e$ and $P^2_e$ can be mapped to appropriate vertices of $H$. 

Note that the red neighbours of each $C_u$ form a clique (whereas all blue vertices of each $C_u$ form an independent set and inner vertices of paths $P^1_e$ and $P^2_e$ have degree~2).
However, as  $(V_1,V_2)$ is an $(r_p,r_q)$-factor cut of $G$, all but at most $r_p$ vertices of these red cliques  have been mapped to $p$ already if
$u\in V_1$ and all but at most $r_q$ vertices have been mapped to $q$ already if $u\in V_2$. By definition of $r_p$ and $r_q$,
 this means that we can map the vertices of the paths $P^1_e$ and $P^2_e$ with $e=uv$ for $u\in V_1$ and $v\in V_2$ to vertices of appropriate shortest paths between $p$ and $q$ in $H$, so that $h$ is a homomorphism from $G^\prime$ to $H$ (recall that we already showed surjectivity). 
In particular, the clique formed by the red neighbours of each $C_u$ is mapped to a clique in $N_p\cup \{p\}$ or $N_q\cup \{q\}$.
  
Now suppose that there exists a surjective homomorphism~$h$ from $G^\prime$ to~$H$.
For a clique $C_u$, we may choose any edge $e$ incident to $u$, such that $C_u^{\prime} = C_u \cup \{g_{u,e}^{\mathrm{r}}\}$ is a clique of size $\omega+1$.
Since $H$ contains no cliques larger than $\omega$, we find that $h$ maps each clique~$C_u^{\prime}$ (which has size $\omega+1$) to a clique in $H$ that contains a reflexive vertex. 
Note that at least two vertices of $C_u^{\prime}$ are mapped to a reflexive vertex. Hence we can define the following partition of $V_G$.
We let $V_1 = \left\{ v \in V_G : p\in h(C_v) \right\}$ and $V_2 = V_G \setminus V_1=\left\{ v \in V_G : q\in h(C_v) \right\}$.
Lemma~\ref{shc:claim1} tells us that $V_1 \ne \emptyset$ and $V_2 \ne \emptyset$.
We define $M = \left\{uv \in E_G  : u \in V_1,\, v \in V_2 \right\}$.

Let $e = uv$ be an arbitrary edge in $M$.
By definition, $h$ maps all of $C_u$ to a clique containing $p$ and all of $C_v$ to a clique containing $q$.
Hence, the vertices of the two paths  $P_e^1$ and $P_e^2$ must be mapped to the vertices of a shortest path between $p$ and $q$.
At most $r_p$ red neighbours of every $C_u$ with $u\in V_1$ can be mapped to a vertex other than $p$. This is because these red neighbours form a clique. As such they must be mapped onto vertices that form a clique in $H$. As such vertices lie on a shortest path from $p$ to $q$, the clique in $H$ has size at most $r_p$.
Similarly, at most $r_q$ red neighbours of every $C_u$ with $u\in V_2$ can be mapped to a vertex other than $q$.
As such, $(V_1,V_2)$ is an $(r_p,r_q)$-factor cut in $G$.
\end{proof}

\subsection{A Small Extension}\label{s-refcl}

Two vertices $u$ and $v$ in a graph $G$ are {\it true twins} if they are adjacent to each other and share the same neighbours in $V_G\setminus \{u,v\}$.
Let $H^{(i,j)}$ be a graph obtained from a connected 2-reflexive graph $H$ with reflexive vertices $p$ and $q$ after introducing $i$ reflexive
true twins of $p$ and $j$ reflexive true twins of $q$.
In the graph $G^\prime$ we increase the cliques $C_u$ to size $\omega+\max(i,j)$. We call the resulting graph~$G''$. Then it is readily seen
that there exists a surjective homomorphism from $G^\prime$ to~$H$ if and only if there exists a surjective homomorphism from $G''$ to
$H^{(i,j)}$.

\begin{theorem}
\label{shc:cor:the_theorem}
For every connected $2$-reflexive graph $H$ and integers $i,j\geq 0$, {\sc Surjective $H^{(i,j)}$-Colouring} is \NP-complete.
\end{theorem}

\section{Target Graphs Of At Most Four Vertices}\label{s-four}

In this section we classify the computational complexity of {\sc Surjective $H$-Colouring} for every target graph~$H$ with at most four vertices.
We require a number of lemmas. 
The first lemma is
proved for compaction and not vertex-surjection. However, the only property of compaction used is vertex-surjection and so it is easy to see it holds in this modified form. The second lemma is also displayed in Figure~\ref{shc:fig:comp-rel}.

\begin{lemma}[\cite{Vi05}]\label{l:2.6}
Let $H$ be a graph with connected components $H_1,\ldots,H_s$. If {\sc Surjective $H_i$-Colouring} is \NP-complete for some $i$, then \shc{} is also \NP-complete.
\end{lemma}

\begin{lemma}[\cite{BKM12}]
\label{l-third}
For every graph $H$, if {\sc $H$-Compaction} is polynomial-time solvable, then \shc{} is polynomial-time solvable.
\end{lemma} 

We also need two results of Golovach, Paulusma and Song. Recall that in our context a tree is a connected graph with no cycles of length at least~3.

\begin{lemma}[\cite{GPS12}]\label{l-gps0}
Let $H$ be an irreflexive non-bipartite graph. Then \shc{} is \NP-complete.
\end{lemma}

\begin{lemma}[\cite{GPS12}]\label{l-gps}
Let $H$ be a tree. Then \shc{} is solvable in polynomial time if $H$ is loop-connected and \NP-complete otherwise.
\end{lemma}

Recall that $C^*_4$ denotes the reflexive cycle on four vertices (see also Figure~\ref{f-3}).

\begin{lemma}[\cite{MP15}]\label{l-c4}
The {\sc Surjective $C^*_4$-Colouring} problem is \NP-complete.
\end{lemma}

\begin{figure*}[ht]
	\centering
	
	\begin{subfigure}[t]{0.24\textwidth}
		\centering
		\begin{tikzpicture}[scale=0.5, baseline=0.36cm]

			\coordinate (a)  at (0, 2);
			\coordinate (b)  at (2, 2);
			\coordinate (c)  at (0, 0);
			\coordinate (d)  at (2, 0);
	
			\fill (a) circle[radius=3pt];
			\fill (b) circle[radius=3pt];
			\fill (c) circle[radius=3pt];
			\fill (d) circle[radius=3pt];
	
			\draw (a) -- (b) -- (d) -- (c) -- (a);
	
			\path[-] (a) edge  [in=165,out=105,loop] node {} ();
			\path[-] (b) edge  [in= 15,out= 75,loop] node {} ();
			\path[-] (c) edge  [in=195,out=255,loop] node {} ();
			\path[-] (d) edge  [in=285,out=345,loop] node {} ();
	
			\path[use as bounding box] (-1.5, -1.5) rectangle (3.5, 3.5);

		\end{tikzpicture}
	\end{subfigure}
	~
	\begin{subfigure}[t]{0.24\textwidth}
		\centering
		\begin{tikzpicture}[scale=0.5, baseline=0.36cm]

			\coordinate (a)  at (0, 2);
			\coordinate (b)  at (2, 2);
			\coordinate (c)  at (0, 0);
			\coordinate (d)  at (2, 0);
	
			\fill (a) circle[radius=3pt];
			\fill (b) circle[radius=3pt];
			\fill (c) circle[radius=3pt];
			\fill (d) circle[radius=3pt];
	
			\draw (a) -- (b) -- (d) -- (c) -- (a);
			\draw (a) -- (d);
	
			\path[use as bounding box] (-1.5, -1.5) rectangle (3.5, 3.5);

		\end{tikzpicture}
	\end{subfigure}
	~
	\begin{subfigure}[t]{0.24\textwidth}
		\centering
		\begin{tikzpicture}[scale=0.5, baseline=0.36cm]

			\coordinate (a)  at ( -1.732, 2);
			\coordinate (b)  at ( -1.732, 0);
			\coordinate (c)  at ( 0, 1);
			\coordinate (d)  at ( 2, 1);
	
			\fill (a) circle[radius=3pt];
			\fill (b) circle[radius=3pt];
			\fill (c) circle[radius=3pt];
			\fill (d) circle[radius=3pt];
	
			\draw (a) -- (b) -- (c) -- (a);
			\draw (c) -- (d);
	
			\path[-] (d) edge  [in=60,out=120,loop] node {} ();
	
			\path[use as bounding box] (-3.232, -1.5) rectangle (3.5, 3.5);

		\end{tikzpicture}
	\end{subfigure}
	
	\caption{The graphs $C_4^{*}$, $D$ and $\mathrm{paw}^{*}$.}\label{f-3}
\end{figure*}
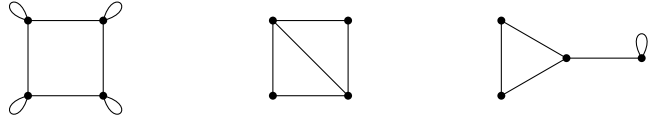

We let $D$ denote the irreflexive diamond, that is, the irreflexive complete graph on four vertices minus an edge.
The (irreflexive) paw is the graph obtained from the triangle after attaching a pendant vertex to one of the vertices of the triangle, that is,
the graph with vertices $x_1$, $x_2$, $y$, $z$ and edges $x_1x_2$, $x_1y$, $x_2y$, $yz$.
We let $\mbox{paw}^*$ denote the graph obtained from the paw after adding a loop to its vertex of degree~1 (that is, following the above notation, the loop $zz$).
Both $D$ and $\mbox{paw}^*$ are displayed in Figure~\ref{f-3} as well. 

\begin{figure*}[ht]
	\centering

	\begin{subfigure}[t]{0.24\textwidth}
		\centering
		\begin{tikzpicture}[scale=0.5, baseline=0.36cm]

			\coordinate (a)  at (0, 2);
			\coordinate (b)  at (2, 2);
			\coordinate (c)  at (0, 0);
			\coordinate (d)  at (2, 0);
	
			\fill (a) circle[radius=3pt];
			\fill (b) circle[radius=3pt];
			\fill (c) circle[radius=3pt];
			\fill (d) circle[radius=3pt];
	
			\draw (a) -- (b) -- (d) -- (c) -- (a);
		
			\path[use as bounding box] (-1.5, -1.5) rectangle (3.5, 3.5);

		\end{tikzpicture}
		\caption{\PP}
	\end{subfigure}
	~
	\begin{subfigure}[t]{0.24\textwidth}
		\centering
		\begin{tikzpicture}[scale=0.5, baseline=0.36cm]

			\coordinate (a)  at (0, 2);
			\coordinate (b)  at (2, 2);
			\coordinate (c)  at (0, 0);
			\coordinate (d)  at (2, 0);
	
			\fill (a) circle[radius=3pt];
			\fill (b) circle[radius=3pt];
			\fill (c) circle[radius=3pt];
			\fill (d) circle[radius=3pt];
	
			\draw (a) -- (b) -- (d) -- (c) -- (a);
	
			\path[-] (a) edge  [in=165,out=105,loop] node {} ();
	
			\path[use as bounding box] (-1.5, -1.5) rectangle (3.5, 3.5);

		\end{tikzpicture}
		\caption{\PP}
	\end{subfigure}
	~
	\begin{subfigure}[t]{0.24\textwidth}
		\centering
		\begin{tikzpicture}[scale=0.5, baseline=0.36cm]

			\coordinate (a)  at (0, 2);
			\coordinate (b)  at (2, 2);
			\coordinate (c)  at (0, 0);
			\coordinate (d)  at (2, 0);
	
			\fill (a) circle[radius=3pt];
			\fill (b) circle[radius=3pt];
			\fill (c) circle[radius=3pt];
			\fill (d) circle[radius=3pt];
	
			\draw (a) -- (b) -- (d) -- (c) -- (a);
	
			\path[-] (a) edge  [in=165,out=105,loop] node {} ();
			\path[-] (b) edge  [in= 15,out= 75,loop] node {} ();
	
			\path[use as bounding box] (-1.5, -1.5) rectangle (3.5, 3.5);

		\end{tikzpicture}
		\caption{\PP}
	\end{subfigure}
	~
	\begin{subfigure}[t]{0.24\textwidth}
		\centering
		\begin{tikzpicture}[scale=0.5, baseline=0.36cm]

			\coordinate (a)  at (0, 2);
			\coordinate (b)  at (2, 2);
			\coordinate (c)  at (0, 0);
			\coordinate (d)  at (2, 0);
	
			\fill (a) circle[radius=3pt];
			\fill (b) circle[radius=3pt];
			\fill (c) circle[radius=3pt];
			\fill (d) circle[radius=3pt];
	
			\draw (a) -- (b) -- (d) -- (c) -- (a);
	
			\path[-] (a) edge  [in=165,out=105,loop] node {} ();
			\path[-] (d) edge  [in=285,out=345,loop] node {} ();
	
			\path[use as bounding box] (-1.5, -1.5) rectangle (3.5, 3.5);

		\end{tikzpicture}
		\caption{\NP-complete}
	\end{subfigure}
	~
	\begin{subfigure}[t]{0.24\textwidth}
		\centering
		\begin{tikzpicture}[scale=0.5, baseline=0.36cm]

			\coordinate (a)  at (0, 2);
			\coordinate (b)  at (2, 2);
			\coordinate (c)  at (0, 0);
			\coordinate (d)  at (2, 0);
	
			\fill (a) circle[radius=3pt];
			\fill (b) circle[radius=3pt];
			\fill (c) circle[radius=3pt];
			\fill (d) circle[radius=3pt];
	
			\draw (a) -- (b) -- (d) -- (c) -- (a);
			
			\path[-] (b) edge  [in= 15,out= 75,loop] node {} ();
			\path[-] (c) edge  [in=195,out=255,loop] node {} ();
			\path[-] (d) edge  [in=285,out=345,loop] node {} ();
	
			\path[use as bounding box] (-1.5, -1.5) rectangle (3.5, 3.5);

		\end{tikzpicture}
		\caption{\PP}
	\end{subfigure}
	~
	\begin{subfigure}[t]{0.24\textwidth}
		\centering
		\begin{tikzpicture}[scale=0.5, baseline=0.36cm]

			\coordinate (a)  at (0, 2);
			\coordinate (b)  at (2, 2);
			\coordinate (c)  at (0, 0);
			\coordinate (d)  at (2, 0);
	
			\fill (a) circle[radius=3pt];
			\fill (b) circle[radius=3pt];
			\fill (c) circle[radius=3pt];
			\fill (d) circle[radius=3pt];
	
			\draw (a) -- (b) -- (d) -- (c) -- (a);
	
			\path[-] (a) edge  [in=165,out=105,loop] node {} ();
			\path[-] (b) edge  [in= 15,out= 75,loop] node {} ();
			\path[-] (c) edge  [in=195,out=255,loop] node {} ();
			\path[-] (d) edge  [in=285,out=345,loop] node {} ();
	
			\path[use as bounding box] (-1.5, -1.5) rectangle (3.5, 3.5);

		\end{tikzpicture}
		\caption{\NP-complete}
	\end{subfigure}
	~
	
	\caption{All cycles $H$ on four vertices.}
	\label{f-cycles}
\end{figure*}
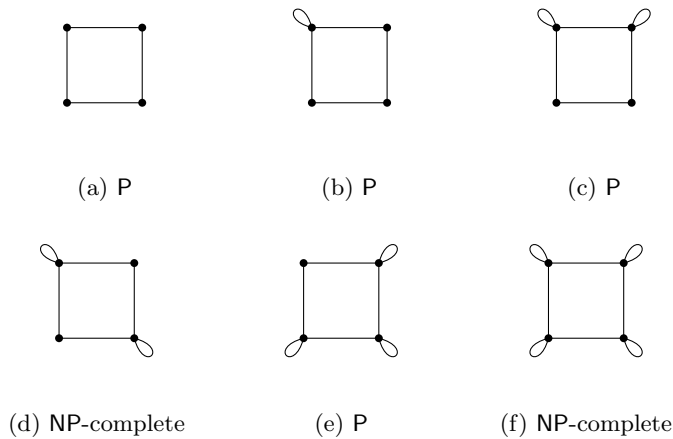

\begin{figure*}[ht]
	\centering

	\begin{subfigure}[t]{0.24\textwidth}
		\centering
		\begin{tikzpicture}[scale=0.5, baseline=0.36cm]

			\coordinate (a)  at (0, 2);
			\coordinate (b)  at (2, 2);
			\coordinate (c)  at (0, 0);
			\coordinate (d)  at (2, 0);
	
			\fill (a) circle[radius=3pt];
			\fill (b) circle[radius=3pt];
			\fill (c) circle[radius=3pt];
			\fill (d) circle[radius=3pt];
	
			\draw (a) -- (b) -- (d) -- (c) -- (a);
			\draw (a) -- (d) -- (c) -- (b);
	
			\path[use as bounding box] (-1.5, -1.5) rectangle (3.5, 3.5);

		\end{tikzpicture}
		\caption{\NP-complete}
	\end{subfigure}
	~
	\begin{subfigure}[t]{0.24\textwidth}
		\centering
		\begin{tikzpicture}[scale=0.5, baseline=0.36cm]

			\coordinate (a)  at (0, 2);
			\coordinate (b)  at (2, 2);
			\coordinate (c)  at (0, 0);
			\coordinate (d)  at (2, 0);
	
			\fill (a) circle[radius=3pt];
			\fill (b) circle[radius=3pt];
			\fill (c) circle[radius=3pt];
			\fill (d) circle[radius=3pt];
	
			\draw (a) -- (b) -- (d) -- (c) -- (a);
			\draw (a) -- (d) -- (c) -- (b);
	
			\path[-] (a) edge  [in=165,out=105,loop] node {} ();
	
			\path[use as bounding box] (-1.5, -1.5) rectangle (3.5, 3.5);

		\end{tikzpicture}
		\caption{\PP}
	\end{subfigure}
	~
	\begin{subfigure}[t]{0.24\textwidth}
		\centering
		\begin{tikzpicture}[scale=0.5, baseline=0.36cm]

			\coordinate (a)  at (0, 2);
			\coordinate (b)  at (2, 2);
			\coordinate (c)  at (0, 0);
			\coordinate (d)  at (2, 0);
	
			\fill (a) circle[radius=3pt];
			\fill (b) circle[radius=3pt];
			\fill (c) circle[radius=3pt];
			\fill (d) circle[radius=3pt];
	
			\draw (a) -- (b) -- (d) -- (c) -- (a);
			\draw (a) -- (d) -- (c) -- (b);
	
			\path[-] (a) edge  [in=165,out=105,loop] node {} ();
			\path[-] (b) edge  [in= 15,out= 75,loop] node {} ();
	
			\path[use as bounding box] (-1.5, -1.5) rectangle (3.5, 3.5);

		\end{tikzpicture}
		\caption{\PP}
	\end{subfigure}
	~
	\begin{subfigure}[t]{0.24\textwidth}
		\centering
		\begin{tikzpicture}[scale=0.5, baseline=0.36cm]

			\coordinate (a)  at (0, 2);
			\coordinate (b)  at (2, 2);
			\coordinate (c)  at (0, 0);
			\coordinate (d)  at (2, 0);
	
			\fill (a) circle[radius=3pt];
			\fill (b) circle[radius=3pt];
			\fill (c) circle[radius=3pt];
			\fill (d) circle[radius=3pt];
	
			\draw (a) -- (b) -- (d) -- (c) -- (a);
			\draw (a) -- (d) -- (c) -- (b);
	
			\path[-] (b) edge  [in= 15,out= 75,loop] node {} ();
			\path[-] (c) edge  [in=195,out=255,loop] node {} ();
			\path[-] (d) edge  [in=285,out=345,loop] node {} ();
	
			\path[use as bounding box] (-1.5, -1.5) rectangle (3.5, 3.5);

		\end{tikzpicture}
		\caption{\PP}
	\end{subfigure}
	~
	\begin{subfigure}[t]{0.24\textwidth}
		\centering
		\begin{tikzpicture}[scale=0.5, baseline=0.36cm]

			\coordinate (a)  at (0, 2);
			\coordinate (b)  at (2, 2);
			\coordinate (c)  at (0, 0);
			\coordinate (d)  at (2, 0);
	
			\fill (a) circle[radius=3pt];
			\fill (b) circle[radius=3pt];
			\fill (c) circle[radius=3pt];
			\fill (d) circle[radius=3pt];
	
			\draw (a) -- (b) -- (d) -- (c) -- (a);
			\draw (a) -- (d) -- (c) -- (b);
	
			\path[-] (a) edge  [in=165,out=105,loop] node {} ();
			\path[-] (b) edge  [in= 15,out= 75,loop] node {} ();
			\path[-] (c) edge  [in=195,out=255,loop] node {} ();
			\path[-] (d) edge  [in=285,out=345,loop] node {} ();
	
			\path[use as bounding box] (-1.5, -1.5) rectangle (3.5, 3.5);

		\end{tikzpicture}
		\caption{\PP}
	\end{subfigure}
	~
	
	\caption{All complete graphs $H$ on four vertices.}
	\label{f-complete}
\end{figure*}
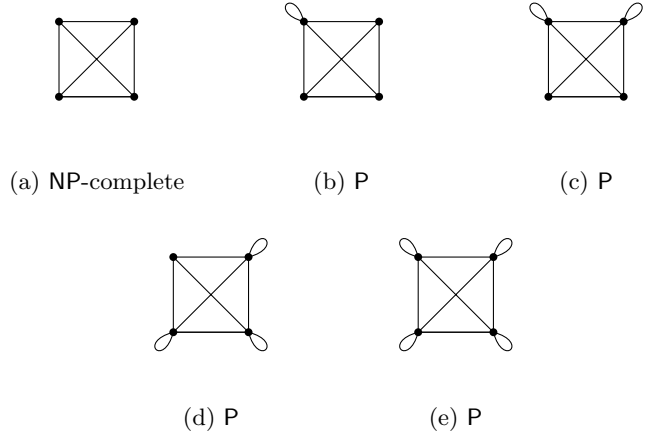

We are now ready to state our main result. 

\begin{theorem}\label{t-dom4}
Let $H$ be a graph with $|V_H|\leq 4$. Then \shc{} is \NP-complete if
some connected component of $H$ is not loop-connected or is an irreflexive complete graph on at least three vertices, or
$H\in \{C_4^*,D,\mathrm{paw}^*\}$. Otherwise \shc{} is polynomial-time solvable.
\end{theorem}

\begin{proof}
Let $H$ be a graph on at most four vertices.
If $H$ is a loop-connected forest (that is, every component of $H$ is loop-connected) or $H$ has a dominating reflexive vertex, 
then Vikas~\cite{Vi05} showed that $H$-{\sc Compaction} is in \PP. Hence, \shc{} is in \PP\ by Lemma~\ref{l-third}.
If $H$ contains a component that is a non-loop-connected tree, then \shc{} is \NP-complete by Lemmas~\ref{l:2.6} and~\ref{l-gps}.
If $H$ is an irreflexive non-bipartite graph, then \shc{} is \NP-complete by Lemma~\ref{l-gps0}.

Note that the above cases cover all graphs $H$ on at most three vertices, all disconnected graphs~$H$ on four vertices and all trees~$H$ on four vertices. The only two graphs $H$ on at most three vertices for which \shc{} is \NP-complete are the irreflexive cycle on three vertices and the 3-vertex path in which the two end-vertices are reflexive. The only disconnected graphs $H$ on four vertices for which \shc{} is \NP-complete are those that contain these two graphs as connected components. The only trees $H$ on four vertices for which \shc{} is \NP-complete are those that are not loop-connected. Hence the theorem holds for every graph~$H$ on at most three vertices, for every disconnected graph~$H$ on four vertices and for every tree~$H$ on four vertices.

From now on we assume that $H$ is a connected graph on four vertices that is not a tree. Then $H$ is either the cycle on four vertices, the complete graph on four vertices, the diamond  or the paw. 
We consider each of these cases separately.

\begin{figure*}[ht]
	\centering

	\begin{subfigure}[t]{0.24\textwidth}
		\centering
		\begin{tikzpicture}[scale=0.5, baseline=0.36cm]

			\coordinate (a)  at (0, 2);
			\coordinate (b)  at (2, 2);
			\coordinate (c)  at (0, 0);
			\coordinate (d)  at (2, 0);
	
			\fill (a) circle[radius=3pt];
			\fill (b) circle[radius=3pt];
			\fill (c) circle[radius=3pt];
			\fill (d) circle[radius=3pt];
	
			\draw (a) -- (b) -- (d) -- (c) -- (a);
			\draw (a) -- (d);
	
			\path[use as bounding box] (-1.5, -1.5) rectangle (3.5, 3.5);

		\end{tikzpicture}
		\caption{\NP-complete}
	\end{subfigure}
	~
	\begin{subfigure}[t]{0.24\textwidth}
		\centering
		\begin{tikzpicture}[scale=0.5, baseline=0.36cm]

			\coordinate (a)  at (0, 2);
			\coordinate (b)  at (2, 2);
			\coordinate (c)  at (0, 0);
			\coordinate (d)  at (2, 0);
	
			\fill (a) circle[radius=3pt];
			\fill (b) circle[radius=3pt];
			\fill (c) circle[radius=3pt];
			\fill (d) circle[radius=3pt];
	
			\draw (a) -- (b) -- (d) -- (c) -- (a);
			\draw (a) -- (d);
	
			\path[-] (a) edge  [in=165,out=105,loop] node {} ();
	
			\path[use as bounding box] (-1.5, -1.5) rectangle (3.5, 3.5);

		\end{tikzpicture}
		\caption{\PP}
	\end{subfigure}
	~
	\begin{subfigure}[t]{0.24\textwidth}
		\centering
		\begin{tikzpicture}[scale=0.5, baseline=0.36cm]

			\coordinate (a)  at (0, 2);
			\coordinate (b)  at (2, 2);
			\coordinate (c)  at (0, 0);
			\coordinate (d)  at (2, 0);
	
			\fill (a) circle[radius=3pt];
			\fill (b) circle[radius=3pt];
			\fill (c) circle[radius=3pt];
			\fill (d) circle[radius=3pt];
	
			\draw (a) -- (b) -- (d) -- (c) -- (a);
			\draw (a) -- (d);
	
			\path[-] (b) edge  [in= 15,out= 75,loop] node {} ();
	
			\path[use as bounding box] (-1.5, -1.5) rectangle (3.5, 3.5);

		\end{tikzpicture}
		\caption{\PP}
	\end{subfigure}
	~
	\begin{subfigure}[t]{0.24\textwidth}
		\centering
		\begin{tikzpicture}[scale=0.5, baseline=0.36cm]

			\coordinate (a)  at (0, 2);
			\coordinate (b)  at (2, 2);
			\coordinate (c)  at (0, 0);
			\coordinate (d)  at (2, 0);
	
			\fill (a) circle[radius=3pt];
			\fill (b) circle[radius=3pt];
			\fill (c) circle[radius=3pt];
			\fill (d) circle[radius=3pt];
	
			\draw (a) -- (b) -- (d) -- (c) -- (a);
			\draw (a) -- (d);
	
			\path[-] (a) edge  [in=165,out=105,loop] node {} ();
			\path[-] (b) edge  [in= 15,out= 75,loop] node {} ();
	
			\path[use as bounding box] (-1.5, -1.5) rectangle (3.5, 3.5);

		\end{tikzpicture}
		\caption{\PP}
	\end{subfigure}
	~
	\begin{subfigure}[t]{0.24\textwidth}
		\centering
		\begin{tikzpicture}[scale=0.5, baseline=0.36cm]

			\coordinate (a)  at (0, 2);
			\coordinate (b)  at (2, 2);
			\coordinate (c)  at (0, 0);
			\coordinate (d)  at (2, 0);
	
			\fill (a) circle[radius=3pt];
			\fill (b) circle[radius=3pt];
			\fill (c) circle[radius=3pt];
			\fill (d) circle[radius=3pt];
	
			\draw (a) -- (b) -- (d) -- (c) -- (a);
			\draw (a) -- (d);
	
			\path[-] (a) edge  [in=165,out=105,loop] node {} ();
			\path[-] (d) edge  [in=285,out=345,loop] node {} ();
	
			\path[use as bounding box] (-1.5, -1.5) rectangle (3.5, 3.5);

		\end{tikzpicture}
		\caption{\PP}
	\end{subfigure}
	~
	\begin{subfigure}[t]{0.24\textwidth}
		\centering
		\begin{tikzpicture}[scale=0.5, baseline=0.36cm]

			\coordinate (a)  at (0, 2);
			\coordinate (b)  at (2, 2);
			\coordinate (c)  at (0, 0);
			\coordinate (d)  at (2, 0);
	
			\fill (a) circle[radius=3pt];
			\fill (b) circle[radius=3pt];
			\fill (c) circle[radius=3pt];
			\fill (d) circle[radius=3pt];
	
			\draw (a) -- (b) -- (d) -- (c) -- (a);
			\draw (a) -- (d);
	
			\path[-] (b) edge  [in= 15,out= 75,loop] node {} ();
			\path[-] (c) edge  [in=195,out=255,loop] node {} ();
	
			\path[use as bounding box] (-1.5, -1.5) rectangle (3.5, 3.5);

		\end{tikzpicture}
		\caption{\NP-complete}
	\end{subfigure}
	~
	\begin{subfigure}[t]{0.24\textwidth}
		\centering
		\begin{tikzpicture}[scale=0.5, baseline=0.36cm]

			\coordinate (a)  at (0, 2);
			\coordinate (b)  at (2, 2);
			\coordinate (c)  at (0, 0);
			\coordinate (d)  at (2, 0);
	
			\fill (a) circle[radius=3pt];
			\fill (b) circle[radius=3pt];
			\fill (c) circle[radius=3pt];
			\fill (d) circle[radius=3pt];
	
			\draw (a) -- (b) -- (d) -- (c) -- (a);
			\draw (a) -- (d);
	
			\path[-] (b) edge  [in= 15,out= 75,loop] node {} ();
			\path[-] (c) edge  [in=195,out=255,loop] node {} ();
			\path[-] (d) edge  [in=285,out=345,loop] node {} ();
	
			\path[use as bounding box] (-1.5, -1.5) rectangle (3.5, 3.5);

		\end{tikzpicture}
		\caption{\PP}
	\end{subfigure}
	~
	\begin{subfigure}[t]{0.24\textwidth}
		\centering
		\begin{tikzpicture}[scale=0.5, baseline=0.36cm]

			\coordinate (a)  at (0, 2);
			\coordinate (b)  at (2, 2);
			\coordinate (c)  at (0, 0);
			\coordinate (d)  at (2, 0);
	
			\fill (a) circle[radius=3pt];
			\fill (b) circle[radius=3pt];
			\fill (c) circle[radius=3pt];
			\fill (d) circle[radius=3pt];
	
			\draw (a) -- (b) -- (d) -- (c) -- (a);
			\draw (a) -- (d);
	
			\path[-] (a) edge  [in=165,out=105,loop] node {} ();
			\path[-] (c) edge  [in=195,out=255,loop] node {} ();
			\path[-] (d) edge  [in=285,out=345,loop] node {} ();
	
			\path[use as bounding box] (-1.5, -1.5) rectangle (3.5, 3.5);

		\end{tikzpicture}
		\caption{\PP}
	\end{subfigure}
~	
	\begin{subfigure}[t]{0.24\textwidth}
		\centering
		\begin{tikzpicture}[scale=0.5, baseline=0.36cm]

			\coordinate (a)  at (0, 2);
			\coordinate (b)  at (2, 2);
			\coordinate (c)  at (0, 0);
			\coordinate (d)  at (2, 0);
	
			\fill (a) circle[radius=3pt];
			\fill (b) circle[radius=3pt];
			\fill (c) circle[radius=3pt];
			\fill (d) circle[radius=3pt];
	
			\draw (a) -- (b) -- (d) -- (c) -- (a);
			\draw (a) -- (d);
	
			\path[-] (a) edge  [in=165,out=105,loop] node {} ();
			\path[-] (b) edge  [in= 15,out= 75,loop] node {} ();
			\path[-] (c) edge  [in=195,out=255,loop] node {} ();
			\path[-] (d) edge  [in=285,out=345,loop] node {} ();
	
			\path[use as bounding box] (-1.5, -1.5) rectangle (3.5, 3.5);

		\end{tikzpicture}
		\caption{\PP}
	\end{subfigure}

	\caption{All diamonds $H$ on four vertices.}
	\label{f-diamonds}
\end{figure*}
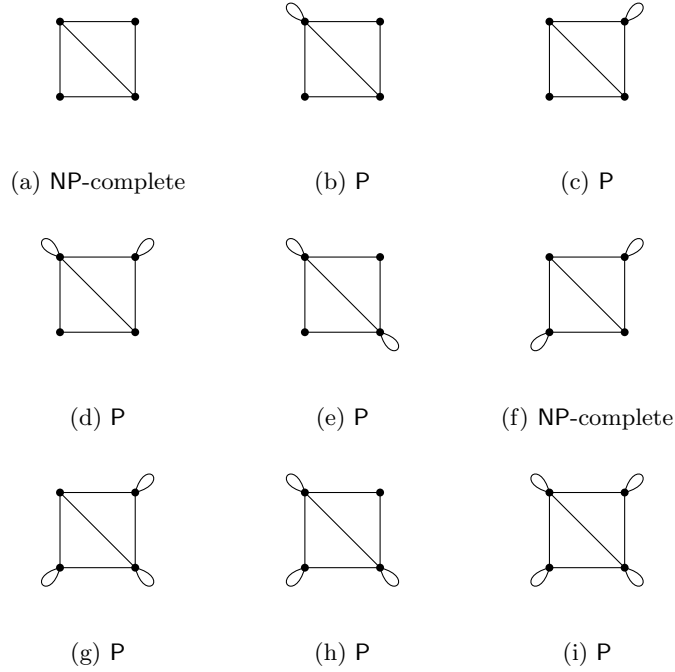

Suppose $H$ is the cycle on four vertices. There are six cases to consider (see also Figure~\ref{f-cycles}).
If $H$ is reflexive, then \shc{} is \NP-complete by Lemma~\ref{l-c4}.
If $H$ is not loop-connected, then $H$ is 2-reflexive, and thus \shc{} is \NP-complete by Theorem~\ref{shc:the_theorem-gt2}.
In the remaining 
four cases $H$ is loop-connected.
For each of these target graphs, Vikas~\cite{Vi05} showed that $H$-{\sc Compaction} is in \PP. Hence, \shc{} is in \PP\ by Lemma~\ref{l-third}. We find that the theorem holds when $H$ is a cycle on four vertices.

Suppose $H$ is the complete graph on four vertices. There are five cases to consider (see also Figure~\ref{f-complete}). If $H$ is irreflexive, then \shc{} is \NP-complete by Lemma~\ref{l-gps0} (as $H$ is non-bipartite as well). 
For each of the other four target graphs, Vikas~\cite{Vi05} showed that $H$-{\sc Compaction} is in \PP. Hence, \shc{} is in \PP\ by Lemma~\ref{l-third}.
We find that the theorem holds when $H$ is the complete graph on four vertices.

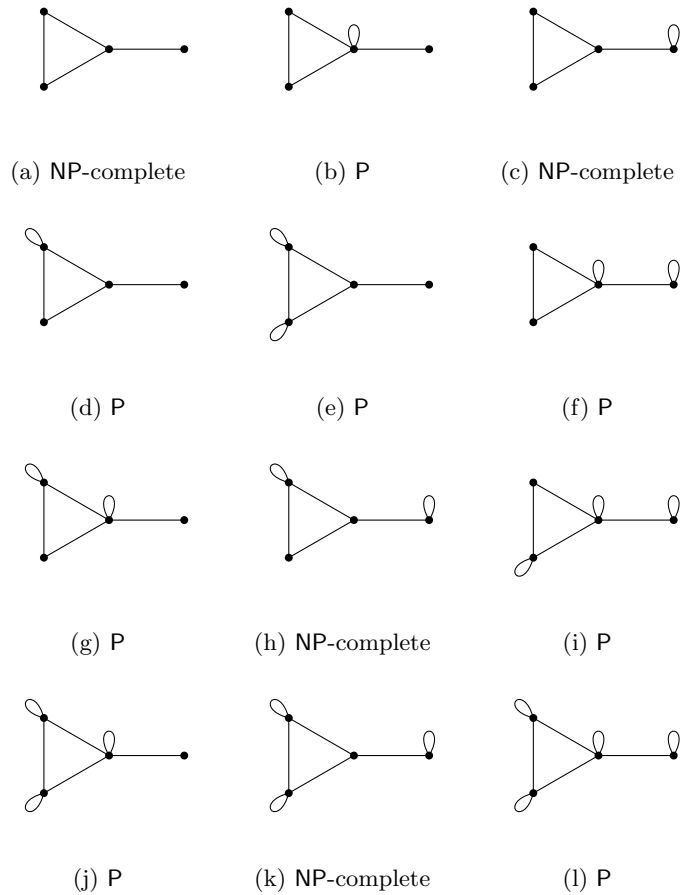
\begin{figure*}[ht]
	\centering

	\begin{subfigure}[t]{0.24\textwidth}
		\centering
		\begin{tikzpicture}[scale=0.5, baseline=0.36cm]

			\coordinate (a)  at ( -1.732, 2);
			\coordinate (b)  at ( -1.732, 0);
			\coordinate (c)  at ( 0, 1);
			\coordinate (d)  at ( 2, 1);
	
			\fill (a) circle[radius=3pt];
			\fill (b) circle[radius=3pt];
			\fill (c) circle[radius=3pt];
			\fill (d) circle[radius=3pt];
	
			\draw (a) -- (b) -- (c) -- (a);
			\draw (c) -- (d);
	
			\path[use as bounding box] (-3.232, -1.5) rectangle (3.5, 3.5);

		\end{tikzpicture}
		\caption{\NP-complete}
	\end{subfigure}
	~
	\begin{subfigure}[t]{0.24\textwidth}
		\centering
		\begin{tikzpicture}[scale=0.5, baseline=0.36cm]

			\coordinate (a)  at ( -1.732, 2);
			\coordinate (b)  at ( -1.732, 0);
			\coordinate (c)  at ( 0, 1);
			\coordinate (d)  at ( 2, 1);
	
			\fill (a) circle[radius=3pt];
			\fill (b) circle[radius=3pt];
			\fill (c) circle[radius=3pt];
			\fill (d) circle[radius=3pt];
	
			\draw (a) -- (b) -- (c) -- (a);
			\draw (c) -- (d);
	
			\path[-] (c) edge  [in=60,out=120,loop] node {} ();
	
			\path[use as bounding box] (-3.232, -1.5) rectangle (3.5, 3.5);

		\end{tikzpicture}
		\caption{\PP}
	\end{subfigure}
	~
	\begin{subfigure}[t]{0.24\textwidth}
		\centering
		\begin{tikzpicture}[scale=0.5, baseline=0.36cm]

			\coordinate (a)  at ( -1.732, 2);
			\coordinate (b)  at ( -1.732, 0);
			\coordinate (c)  at ( 0, 1);
			\coordinate (d)  at ( 2, 1);
	
			\fill (a) circle[radius=3pt];
			\fill (b) circle[radius=3pt];
			\fill (c) circle[radius=3pt];
			\fill (d) circle[radius=3pt];
	
			\draw (a) -- (b) -- (c) -- (a);
			\draw (c) -- (d);
	
			\path[-] (d) edge  [in=60,out=120,loop] node {} ();
	
			\path[use as bounding box] (-3.232, -1.5) rectangle (3.5, 3.5);

		\end{tikzpicture}
		\caption{\NP-complete}
	\end{subfigure}
	~
	\begin{subfigure}[t]{0.24\textwidth}
		\centering
		\begin{tikzpicture}[scale=0.5, baseline=0.36cm]

			\coordinate (a)  at ( -1.732, 2);
			\coordinate (b)  at ( -1.732, 0);
			\coordinate (c)  at ( 0, 1);
			\coordinate (d)  at ( 2, 1);
	
			\fill (a) circle[radius=3pt];
			\fill (b) circle[radius=3pt];
			\fill (c) circle[radius=3pt];
			\fill (d) circle[radius=3pt];
	
			\draw (a) -- (b) -- (c) -- (a);
			\draw (c) -- (d);
	
			\path[-] (a) edge  [in=165,out=105,loop] node {} ();
	
			\path[use as bounding box] (-3.232, -1.5) rectangle (3.5, 3.5);

		\end{tikzpicture}
		\caption{\PP}
	\end{subfigure}
	~	
	\begin{subfigure}[t]{0.24\textwidth}
		\centering
		\begin{tikzpicture}[scale=0.5, baseline=0.36cm]

			\coordinate (a)  at ( -1.732, 2);
			\coordinate (b)  at ( -1.732, 0);
			\coordinate (c)  at ( 0, 1);
			\coordinate (d)  at ( 2, 1);
	
			\fill (a) circle[radius=3pt];
			\fill (b) circle[radius=3pt];
			\fill (c) circle[radius=3pt];
			\fill (d) circle[radius=3pt];
	
			\draw (a) -- (b) -- (c) -- (a);
			\draw (c) -- (d);
	
			\path[-] (a) edge  [in=165,out=105,loop] node {} ();
			\path[-] (b) edge  [in=195,out=255,loop] node {} ();
	
			\path[use as bounding box] (-3.232, -1.5) rectangle (3.5, 3.5);

		\end{tikzpicture}
		\caption{\PP}
	\end{subfigure}
	~
	\begin{subfigure}[t]{0.24\textwidth}
		\centering
		\begin{tikzpicture}[scale=0.5, baseline=0.36cm]

			\coordinate (a)  at ( -1.732, 2);
			\coordinate (b)  at ( -1.732, 0);
			\coordinate (c)  at ( 0, 1);
			\coordinate (d)  at ( 2, 1);
	
			\fill (a) circle[radius=3pt];
			\fill (b) circle[radius=3pt];
			\fill (c) circle[radius=3pt];
			\fill (d) circle[radius=3pt];
	
			\draw (a) -- (b) -- (c) -- (a);
			\draw (c) -- (d);
	
			\path[-] (c) edge  [in=60,out=120,loop] node {} ();
			\path[-] (d) edge  [in=60,out=120,loop] node {} ();
	
			\path[use as bounding box] (-3.232, -1.5) rectangle (3.5, 3.5);

		\end{tikzpicture}
		\caption{\PP}
	\end{subfigure}
	~
	\begin{subfigure}[t]{0.24\textwidth}
		\centering
		\begin{tikzpicture}[scale=0.5, baseline=0.36cm]

			\coordinate (a)  at ( -1.732, 2);
			\coordinate (b)  at ( -1.732, 0);
			\coordinate (c)  at ( 0, 1);
			\coordinate (d)  at ( 2, 1);
	
			\fill (a) circle[radius=3pt];
			\fill (b) circle[radius=3pt];
			\fill (c) circle[radius=3pt];
			\fill (d) circle[radius=3pt];
	
			\draw (a) -- (b) -- (c) -- (a);
			\draw (c) -- (d);
	
			\path[-] (a) edge  [in=165,out=105,loop] node {} ();
			\path[-] (c) edge  [in=60,out=120,loop] node {} ();
	
			\path[use as bounding box] (-3.232, -1.5) rectangle (3.5, 3.5);

		\end{tikzpicture}
		\caption{\PP}
	\end{subfigure}
	~
	\begin{subfigure}[t]{0.24\textwidth}
		\centering
		\begin{tikzpicture}[scale=0.5, baseline=0.36cm]

			\coordinate (a)  at ( -1.732, 2);
			\coordinate (b)  at ( -1.732, 0);
			\coordinate (c)  at ( 0, 1);
			\coordinate (d)  at ( 2, 1);
	
			\fill (a) circle[radius=3pt];
			\fill (b) circle[radius=3pt];
			\fill (c) circle[radius=3pt];
			\fill (d) circle[radius=3pt];
	
			\draw (a) -- (b) -- (c) -- (a);
			\draw (c) -- (d);
	
			\path[-] (a) edge  [in=165,out=105,loop] node {} ();
			\path[-] (d) edge  [in=60,out=120,loop] node {} ();
	
			\path[use as bounding box] (-3.232, -1.5) rectangle (3.5, 3.5);

		\end{tikzpicture}
		\caption{\NP-complete}
	\end{subfigure}
	~
	\begin{subfigure}[t]{0.24\textwidth}
		\centering
		\begin{tikzpicture}[scale=0.5, baseline=0.36cm]

			\coordinate (a)  at ( -1.732, 2);
			\coordinate (b)  at ( -1.732, 0);
			\coordinate (c)  at ( 0, 1);
			\coordinate (d)  at ( 2, 1);
	
			\fill (a) circle[radius=3pt];
			\fill (b) circle[radius=3pt];
			\fill (c) circle[radius=3pt];
			\fill (d) circle[radius=3pt];
	
			\draw (a) -- (b) -- (c) -- (a);
			\draw (c) -- (d);

			\path[-] (b) edge  [in=195,out=255,loop] node {} ();
			\path[-] (c) edge  [in=60,out=120,loop] node {} ();
			\path[-] (d) edge  [in=60,out=120,loop] node {} ();
	
			\path[use as bounding box] (-3.232, -1.5) rectangle (3.5, 3.5);

		\end{tikzpicture}
		\caption{\PP}
	\end{subfigure}
	~
	\begin{subfigure}[t]{0.24\textwidth}
		\centering
		\begin{tikzpicture}[scale=0.5, baseline=0.36cm]

			\coordinate (a)  at ( -1.732, 2);
			\coordinate (b)  at ( -1.732, 0);
			\coordinate (c)  at ( 0, 1);
			\coordinate (d)  at ( 2, 1);
	
			\fill (a) circle[radius=3pt];
			\fill (b) circle[radius=3pt];
			\fill (c) circle[radius=3pt];
			\fill (d) circle[radius=3pt];
	
			\draw (a) -- (b) -- (c) -- (a);
			\draw (c) -- (d);
	
			\path[-] (a) edge  [in=165,out=105,loop] node {} ();
			\path[-] (b) edge  [in=195,out=255,loop] node {} ();
			\path[-] (c) edge  [in=60,out=120,loop] node {} ();
	
			\path[use as bounding box] (-3.232, -1.5) rectangle (3.5, 3.5);

		\end{tikzpicture}
		\caption{\PP}
	\end{subfigure}
	~
	\begin{subfigure}[t]{0.24\textwidth}
		\centering
		\begin{tikzpicture}[scale=0.5, baseline=0.36cm]

			\coordinate (a)  at ( -1.732, 2);
			\coordinate (b)  at ( -1.732, 0);
			\coordinate (c)  at ( 0, 1);
			\coordinate (d)  at ( 2, 1);
	
			\fill (a) circle[radius=3pt];
			\fill (b) circle[radius=3pt];
			\fill (c) circle[radius=3pt];
			\fill (d) circle[radius=3pt];
	
			\draw (a) -- (b) -- (c) -- (a);
			\draw (c) -- (d);
	
			\path[-] (a) edge  [in=165,out=105,loop] node {} ();
			\path[-] (b) edge  [in=195,out=255,loop] node {} ();
			\path[-] (d) edge  [in=60,out=120,loop] node {} ();
	
			\path[use as bounding box] (-3.232, -1.5) rectangle (3.5, 3.5);

		\end{tikzpicture}
		\caption{\NP-complete}
	\end{subfigure}
	~
	\begin{subfigure}[t]{0.24\textwidth}
		\centering
		\begin{tikzpicture}[scale=0.5, baseline=0.36cm]

			\coordinate (a)  at ( -1.732, 2);
			\coordinate (b)  at ( -1.732, 0);
			\coordinate (c)  at ( 0, 1);
			\coordinate (d)  at ( 2, 1);
	
			\fill (a) circle[radius=3pt];
			\fill (b) circle[radius=3pt];
			\fill (c) circle[radius=3pt];
			\fill (d) circle[radius=3pt];
	
			\draw (a) -- (b) -- (c) -- (a);
			\draw (c) -- (d);
	
			\path[-] (a) edge  [in=165,out=105,loop] node {} ();
			\path[-] (b) edge  [in=195,out=255,loop] node {} ();
			\path[-] (c) edge  [in=60,out=120,loop] node {} ();
			\path[-] (d) edge  [in=60,out=120,loop] node {} ();
	
			\path[use as bounding box] (-3.232, -1.5) rectangle (3.5, 3.5);

		\end{tikzpicture}
		\caption{\PP}
	\end{subfigure}

	\caption{All paws $H$ on four vertices.}
	\label{f-paws}
\end{figure*}

Suppose $H$ is the diamond. There are nine cases to consider (see also Figure~\ref{f-diamonds}).  If $H$ is irreflexive, then \shc{} is \NP-complete by Lemma~\ref{l-gps0} (as $H$ is non-bipartite as well). 
If $H$ is not loop-connected, then $H$ is 2-reflexive, and thus \shc{} is \NP-complete by Theorem~\ref{shc:the_theorem-gt2}. 
For the remaining seven target graphs, Vikas~\cite{Vi05} showed that $H$-{\sc Compaction} is in \PP. Hence, \shc{} is in \PP\ by Lemma~\ref{l-third}. We find that the theorem holds when $H$ is the diamond.

Suppose $H$ is the paw with vertices $x_1,x_2,y,z$ and edges $x_1x_2$, $x_1y$, $x_2y$ and $yz$ and possibly one or more loops.
There are twelve cases to consider (see also Figure~\ref{f-paws}).  If $H$ is irreflexive, then \shc{} is \NP-complete by Lemma~\ref{l-gps0} (as $H$ is non-bipartite as well). 
If $H$ is not loop-connected, then the set of reflexive vertices is formed by one or two vertices from $\{x_1,x_2\}$ and $z$.
Then \shc{} is \NP-complete by Theorem~\ref{shc:cor:the_theorem}. We are left with 
nine cases.
Vikas~\cite{Vi05} showed that  $H$-{\sc Compaction} is in \PP\ 
for all of these cases except for the case where $z$ is the only reflexive vertex.
Hence, for eight of these nine cases, \shc{} is in \PP\ by Lemma~\ref{l-third}. 

We are left to consider the case in which $z$ is the (only) reflexive vertex. Recall that we denote this target by $\mathrm{paw}^*$.
Theorem 3.5 of \cite{Vi05} proves that {\sc $\mathrm{paw}^*$-Compaction} is \NP-complete using a reduction from $C_3$-{\sc Retraction} (which is \NP-complete), but we will argue the proof works also for {\sc Surjective $\mathrm{paw}^*$-Colouring}.
It is shown that (i) a graph $G$ retracts to $C_3$ if and only if a certain graph $G'$ retracts to $\mathrm{paw}^*$ if and only if (iii) $G'$ compacts to $\mathrm{paw}^*$.
The salient part of the proof is Lemma~3.5.2 of~\cite{Vi05}, in which it is argued that (ii) and (iii) are equivalent. We note that if a graph retracts to another graph, then there exists a surjective homomorphism from the first graph to the second graph.
Hence, we need to verify only whether $G'$ retracts to $\mathrm{paw}^*$ should there exist a surjective homomorphism from $G'$ to $\mathrm{paw}^*$.
In the proof of Lemma~3.5.2 of~\cite{Vi05}, the properties of compaction are only used three times.
The first two are paragraph 2, line 2 and paragraph 7, line 4 (in the proof of Lemma 3.5.2).
The only property used of compaction on these two occasions is vertex surjection.
Finally, compaction is alluded to in the final paragraph of the proof, but here any homomorphism would have the desired property.
Thus, Vikas~\cite{Vi05} has actually proved that  $G'$ retracts to $\mathrm{paw}^*$ if and only if $G'$ has a surjective homomorphism to $\mathrm{paw}^*$, and it follows
that  {\sc Surjective $\mathrm{paw}^*$-Colouring} is \NP-complete. 

From the above we conclude that the theorem holds in all cases when $H$ is the paw. This completes the proof of 
Theorem~\ref{t-dom4}.
\end{proof}
Theorem~\ref{t-dom4} corresponds to Vikas' complexity classification of {\sc $H$-Compaction} for targets graphs~$H$ of at most four vertices.  Vikas~\cite{Vi05} showed that  {\sc $H$-Compaction} and {\sc $H$-Retraction} are polynomially equivalent for target graphs~$H$ of at most four vertices. 
Thus, we obtain the following corollary. 

\begin{corollary}\label{c:dom4}
Let $H$ be a 
graph on at most four vertices.
Then the three problems \shc{}, {\sc $H$-Compaction} and {\sc $H$-Retraction} 
are polynomially equivalent.
\end{corollary}

\end{document}